\documentclass[fleqn,usenatbib]{mnras}

\usepackage[T1]{fontenc}
\usepackage{ae,aecompl}
\usepackage[utf8]{inputenc}
\usepackage[english]{babel}
\usepackage{rotating}
\usepackage{multirow}
\usepackage{hyperref}

\newtheorem{theorem}{Theorem}
\newtheorem{corollary}{Corollary}

\usepackage{graphicx}	
\usepackage{amsmath}	
\usepackage{amssymb,latexsym}	
\usepackage{verbatim}
\usepackage{subfig}
\usepackage{array}
\newcolumntype{P}[1]{>{\centering\arraybackslash}p{#1}}

\usepackage{mathtools}
\usepackage{xcolor}

\DeclareMathOperator*{\argmax}{arg\,max}
\definecolor{Gr1}{HTML}{1B9E77}
\definecolor{Gr2}{HTML}{D95F02}
\definecolor{Gr3}{HTML}{7570B3}
\definecolor{Gr4}{HTML}{E7298A}
\definecolor{Gr5}{HTML}{66A61E}

\newcommand{\mR}{\mathcal R}

\title[Five Distinct Kinds of Gamma Ray Bursts]{Multivariate
  $t$-Mixtures-Model-based Cluster Analysis of BATSE Catalog
  Establishes Importance of All  Observed Parameters, Confirms Five
  Distinct Ellipsoidal Sub-populations of Gamma Ray Bursts}

\author[Chattopadhyay and Maitra]{
  Souradeep Chattopadhyay,$^{1}$
  and Ranjan Maitra,$^{1}$\thanks{E-mail: maitra@iastate.edu (RM)}\\
$^{1}$ Department of Statistics, Iowa State University, 2438, Osborn
  Drive, Ames, Iowa 50011-1090, USA
}
\date{Accepted XXX. Received YYY; in original form ZZZ}

\pubyear{2017}

\begin{document}
\label{firstpage}
\pagerange{\pageref{firstpage}--\pageref{lastpage}}
\maketitle

\begin{abstract}
Determining the kinds of gamma-ray bursts (GRBs) has been of  interest to
astronomers for many years. We  analyzed 1599 GRBs from the Burst and Transient Source Experiment (BATSE) 4Br catalogue using  $t$-mixtures-model-based clustering on
all  nine observed parameters ($T_{50}$, $T_{90}$, $F_1$, $F_2$,
$F_3$, $F_4$, $P_{64}$, $P_{256}$, $P_{1024}$) and found evidence of five
types of GRBs. Our results further refine the
findings of  \citet{chattopadhyayandmaitra17} by providing groups that
are more distinct. Using the \citet{mukherjeeetal98}  classification scheme (also used by \citet{chattopadhyayandmaitra17}) of duration, total fluence ($F_t = F_1 + F_2 + F_3 + F_4$)) and spectrum (using Hardness Ratio $H_{321} = F_3/(F_1 + F_2)$)  our five groups are classified as long-intermediate-intermediate, short-faint-intermediate, short-faint-soft, long-bright-hard, and long-intermediate-hard. We also classify 374 GRBs in the
BATSE catalogue that have incomplete information in some of the observed
variables (mainly the four time integrated fluences $F_1$, $F_2$,
$F_3$ and $F_4$) to the five groups obtained, using the 1599 GRBs
having complete information in all the observed variables. Our classification scheme puts 138 GRBs in the first group, 52 GRBs in the second group, 33 GRBs in the third group, 127 GRBs in the fourth group and 24 GRBs in the fifth group.
\end{abstract}




\section{Introduction}
Gamma Ray Bursts (GRBs) are among the brightest electromagnetic events
known in space and have been actively researched ever since their
discovery in the late 1960s, mainly because researchers hypothesize
that these celestial events hold the clue to understanding of
numerous mysteries of the outer cosmos. The source and nature of these
highly explosive events remain unresolved
~\citep{chattopadhyayetal07} with researchers hypothesizing that GRBs
are a heterogeneous group of several subpopulations \citep[{\em 
  e.g.,}][]{mazetsetal81,norrisetal84,dezalayetal92} but there are 
questions on the number of these groups and their underlying
properties. Most analyses pertaining to GRBs have been carried out
using duration variable $T_{90}$ (or the time by which 90 per cent of the
flux arrive) while a few analysts have used fluence and spectral properties along with duration. 
\cite{kouveliotouetal93} found that the $\log_{10}T_{90}$ variable
from the Burst and Transient Source Experiment (BATSE) Catalogue
follows a bimodal distribution and established the two well-known
classes of GRBs, namely the short duration  ($T_{90} < 2s $) and the
long duration bursts ($T_{90} > 2s$). The progenitors
of short duration bursts are thought to be merger of two neutron stars
(NS-NS) or that of a neutron star with a black hole (NS-BH) \citep{nakar07},
while that of long duration bursts are largely believed to be associated with
the collapse of massive
stars~\citep{paczynski98,woosleyandbloom06}. Many other authors 
subsequently carried out several experimental studies using BATSE and
other catalogues and reported a variety of
findings. \citet{pendletonetal97} used 882 GRBs from the BATSE catalog
to perform spectral analysis and found two classes of bursts -- the
high Energy  (HE) and the non high Energy  (NHE) Bursts. \citet{horvathetal98} proposed the presence of a third class
of GRBs by making two and three Gaussian components fits to the $\log_{10}T_{90}$
variable of 797 GRBs in the BATSE 3B catalog. Several
authors 
\citep{horvath02,horvathetal08,horvath09,tarnopolski15,horvathandtoth16,zitounietal15,hujaetal09,horvathandtoth16} 
have since then supported the presence of a third Gaussian
component but \citet{zhangetal16} and \citet{kulkarnianddesai17} have
concluded that the duration variables show  a three-Gaussian-components model only for the \textit{Swift}/ BAT 
dataset but a two-Gaussian-components model for the BATSE and
\textit{Fermi} data sets. Recently \citet{acunerandryde18} found 
five types of GRBs in the {\it Fermi} catalogue. 
\citet{ripaetal09} analyzed the duration and
hardness ratios of 427 GRBs from the RHESSI satellite and found that
a $\chi^2$- or $F$- test on $T_{90}$ does not indicate any
statistically significant intermediate group in the RHESSI data but a
maximum likelihood test using $T_{90}$ and hardness indicates a
statistically significant intermediate group in the same data
set. They concluded that like BATSE, RHESSI also shows evidence of the
presence of an intermediate group. However, use of a $\chi^2$-test on
twice the difference in log likelihoods between two models 
assumes that the larger model is nested within the null model, an
assumption that generally does not hold for non-hierarchical
clustering algorithms
\citep{chattopadhyayandmaitra17}. \citep[Refer to][for a review of testing
mechanisms in such situations.]{maitraetal12}
 
 \citet{mukherjeeetal98} first considered multivariate analysis by
 carrying out non-parametric hierarchical clustering, using six
 variables on 797 GRBs from the BATSE 3B catalog and found evidence of
 three groups. They also performed Model Based Clustering (MBC) by
 eliminating three of those six variables citing presence of
 redundancy through visual inspection. \citet{chattopadhyayetal07}
 carried out $k$-means using the same six variables used by
 \citet{mukherjeeetal98} for non-parametric hierarchical clustering
 and supported the presence of three groups in the BATSE 4B catalog
 \citep[but see][for caution on the use of $k$-means for BATSE
 data.]{chattopadhyayandmaitra17} \citet{chattopadhyayandmaitra17}
 carried out model-based variable selection on the six variables used
 by \citet{chattopadhyayetal07} and \citet{mukherjeeetal98} and were unable to find any evidence of redundancy among them. They carried out
 MBC using Gaussian mixtures, using the same six variables and obtained
 five elliptically-dispersed groups. The BATSE 4Br catalog has a number of zero
 entries in some of the observed variables, mostly in the time
 integrated fluences $F_1-F_4$. Citing personal communication from
 Charles Meegan, \citet{chattopadhyayandmaitra17} pointed out that
 these zero values are not numerical zeroes but missing parameter
 readings on 
 a GRB and hence including them as numerical values in the analysis is
 inappropriate because of the potential for bias in the results. Most
 authors performing multivariate analysis have removed 
 those GRBs that have incomplete information in them, as a result of
 which their properties have never been extensively studied. In this
 paper we have attempted to study the properties of these bursts
 having incomplete information by classifying them to groups obtained
 using GRBs having complete information on all observed variables. 
 
 Clustering is an unsupervised learning approach to group observations
 without any response variable. Clustering 
 algorithms are broadly of the  hierarchical and
 the non-hierarchical types. The former consists of 
 both agglomerative and divisive algorithms where groups are formed in
 a tree-like hierarchy with the property that observations that are
 together at one level are also together higher up  the
 tree. Non-hierarchical algorithms, such as MBC or $k$-means, typically 
 optimize an objective function using iterative greedy algorithms for
 a specified number of groups. The objective function is often
 multimodal and requires careful initialization~\citep{maitra09}. For a
 detailed review on clustering see \citet{chattopadhyayandmaitra17}.

The work of \citet{chattopadhyayandmaitra17} carefully analyzed
the BATSE 4Br data using the two duration variables, the peak flux in
time bins of 256 milliseconds (ms), the total fluence and two spectral hardness
measures $H_{32}$ and $H_{321}$. They established  five
Gaussian-dispersed groups but upset the 
existing view in the astrophysical community that there are between
two and three kinds of GRBs in the BATSE 4Br catalogue. This led us to
wonder if there really were fewer than five ellipsoidal groups that our
methods were not picking up because the assumption of Gaussian
components was not allowing for heavier-tailed, wider-dispersed
groups such as described by the multivariate $t$-distribution. We
also wondered if the other parameters in the BATSE 4Br catalog
summarily discarded or summarized by other authors contained important 
clustering information that would help in  arriving at better-defined
groups. In this article, we therefore analyze whether all nine observed
parameters are needed in clustering the GRBs in the catalogue. We also
examine MBC on the GRB dataset using
multivariate $t$-mixtures.

The remainder of the article is structured as follows. 
 Section  \ref{ov:background}  provides an overview of MBC and
 classification using mixtures of multivariate $t$-distributions that
 allow for more general model-based representations of elliptical
 subpopulations than do Gaussian mixtures. 
Section \ref{GRB:1599:full} establishes that all nine original parameters have
 clustering information, and analyses and discusses results on
$t$-mixtures-MBC ($t$MMBC) done on the 1599
 BATSE 
 4Br GRBs having observations on all nine parameters. Finally Section
 \ref{class-374} classifies the GRBs  having incomplete information to
 the groups obtained by MBC on the  1599 GRBs and examines their
 properties. We conclude with some discussion in Section \ref{conclusion}.

 \section{Overview of MBC and Classification}
 \label{ov:background}
 We briefly describe $t$MMBC and Classification, specifically including methods and techniques that are easily implemented  using the open-source statistical
software R~\citep{R} and its packages. 
 \subsection{Preliminaries}
 \label{ov:teigen}
 \subsubsection{The Multivariate $t$ distribution}
 Let $Y$ be a $p$-dimensional random vector having the multivariate
 Gaussian distribution $N_{p}(0, \Sigma)$ and $S$ be a random
 variable, independent of $Y$, that has a  $\chi^{2}$ distribution
 with ${\nu}$ degrees of freedom. Let $\Sigma$ be a
 positive-definite matrix. Then $X = \mu + Y{\sqrt{\nu/S}}$ follows a multivariate $t_p(\mu,\Sigma;\nu)$ distribution with
 mean vector $\mu$, scale matrix $\Sigma$ and degrees of freedom
 $\nu$, and multivariate probability density function (PDF)
 \begin{equation}
 \label{mult:t}
f_t(x; \mu, \Sigma, \nu) = \frac{\Gamma (\nu +
  p)/2}{\Gamma(\nu/2){\nu}^{\frac{p}{2}}{\pi}^{\frac{p}{2}}{|\Sigma|}^\frac{1}{2}}\times
\Big[1+ \frac{1}{\nu}(X-\mu)^T\Sigma^{-1}(X-\mu)\Big], 
 \end{equation}
for $x\in\mathbb{R}^p$.
The multivariate $t$-distribution is centered and ellipsoidally
symmetric around its mean vector $E(X)=\mu$. The variance-covariance
(or dispersion) matrix is given by $Var(X) = \nu\Sigma/(\nu - 2)$,
therefore having higher spread than the $N_p(\mu,\Sigma)$
distribution, with the exact amount of spread modulated by the degrees
of freedom $\nu$. It is easy to see 
that as $\nu\rightarrow\infty$, the dispersion converges to
$\Sigma$ and indeed $t_p(\mu,\Sigma;\nu)$ converges in law (distribution) to
$N_p(\mu,\Sigma)$. We illustrate  the influence of $\nu$ through a set
of two-dimensional examples in Fig. \ref{coutour.t} that displays
the contour density plot of 
multivariate $t_\nu$ distributions for $\nu=5, 15, 25,\infty$. The
contours are for the ellipses of concentration that 
contain the densest $100\alpha$ per cent of the distribution for $\alpha =
0.1,0.2,\ldots, 0.9, 0.99$. 
The figure corresponding to $\nu=5$ has
the highest spread, with more observations in the tails, but this
spread and tail-preference of the observations decreases with
increasing $\nu$. With infinite degrees of freedom, the multivariate
$t$-density has a similar spread as the  multivariate normal
density. These example illustrate how the multivariate
$t_\nu$ density is concentrated or dispersed around the mean vector $\mu$
accordingly as $\nu$ increases or decreases. 
 \begin{figure}
\mbox{ 
\subfloat[$\nu=5$]{\label{fig:cont:5}\includegraphics[width=0.25\textwidth]{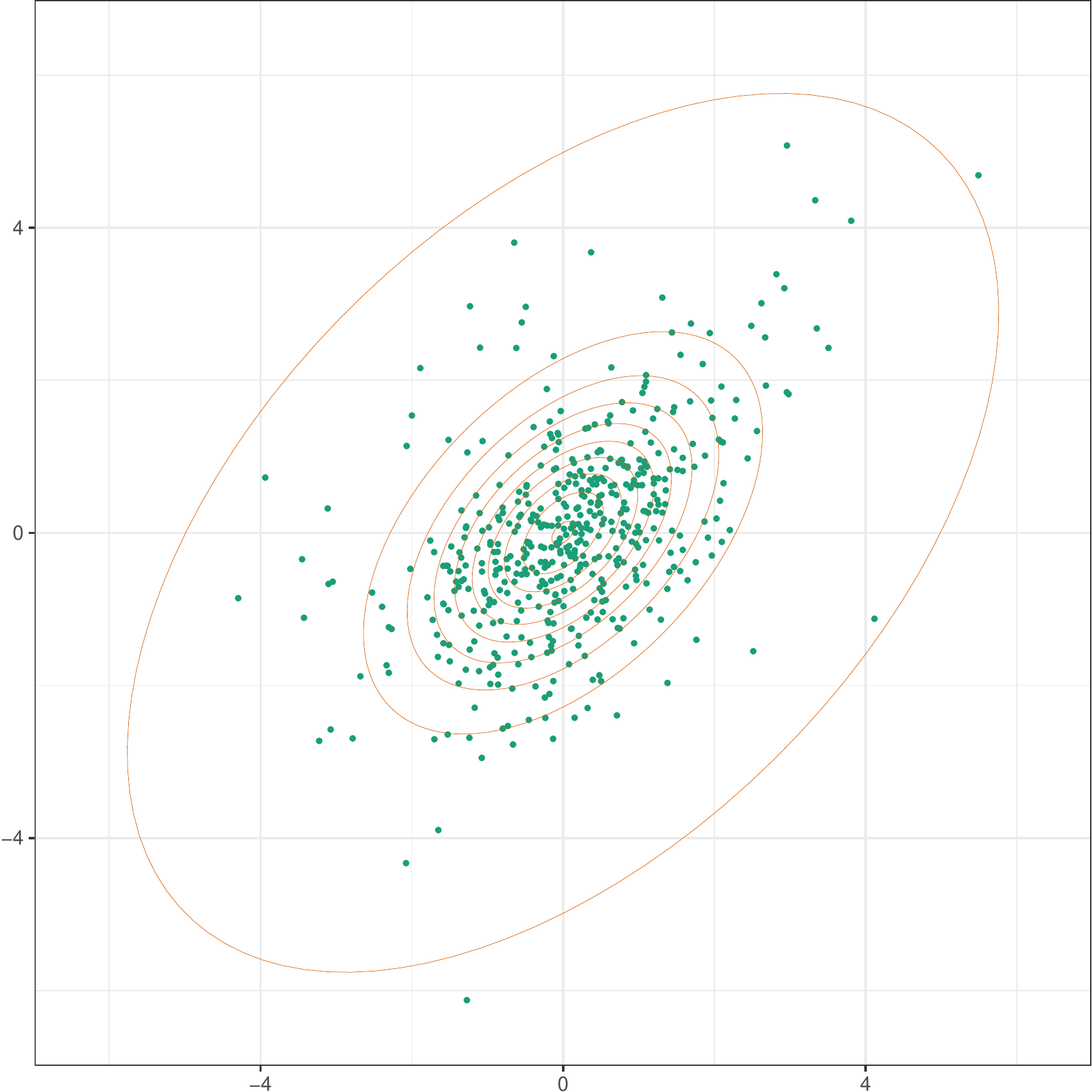}}
\subfloat[$\nu=15$]{\label{fig:cont:15}\includegraphics[width=0.25\textwidth]{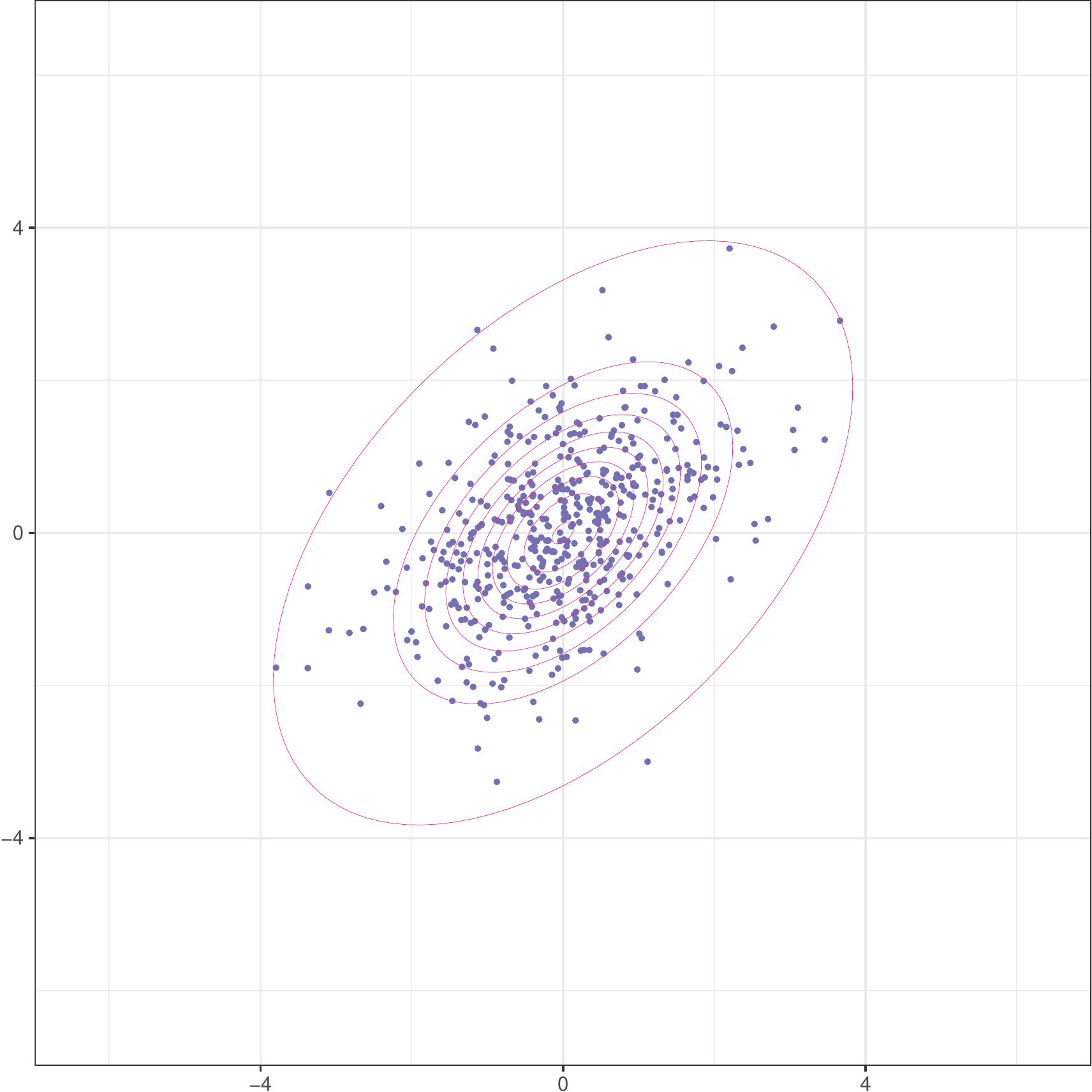}}}
\mbox{
\subfloat[$\nu=25$]{\label{fig:cont:25}\includegraphics[width=0.25\textwidth]{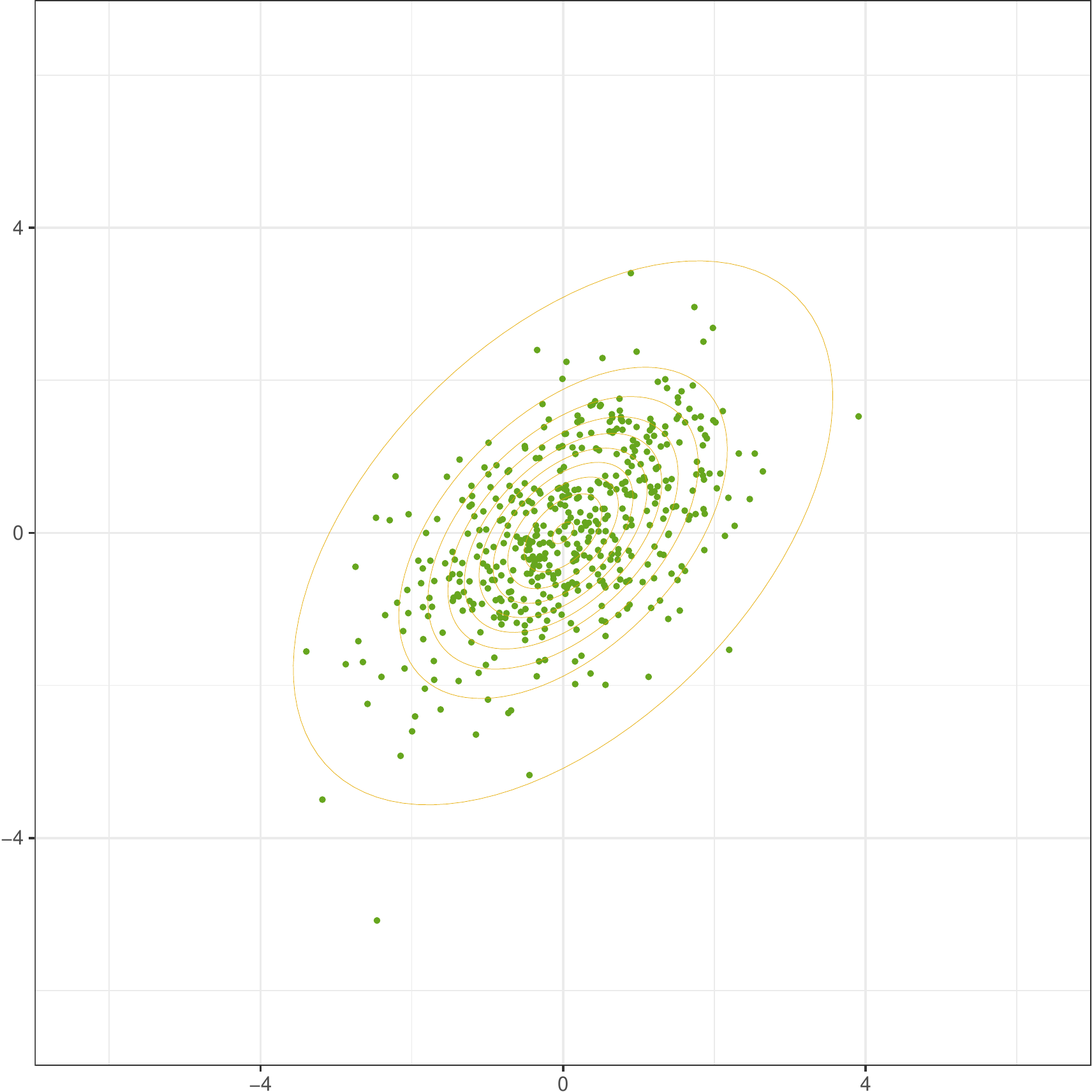}}
\subfloat[$\nu=\infty$]{\label{fig:cont:inf}\includegraphics[width=0.25\textwidth]{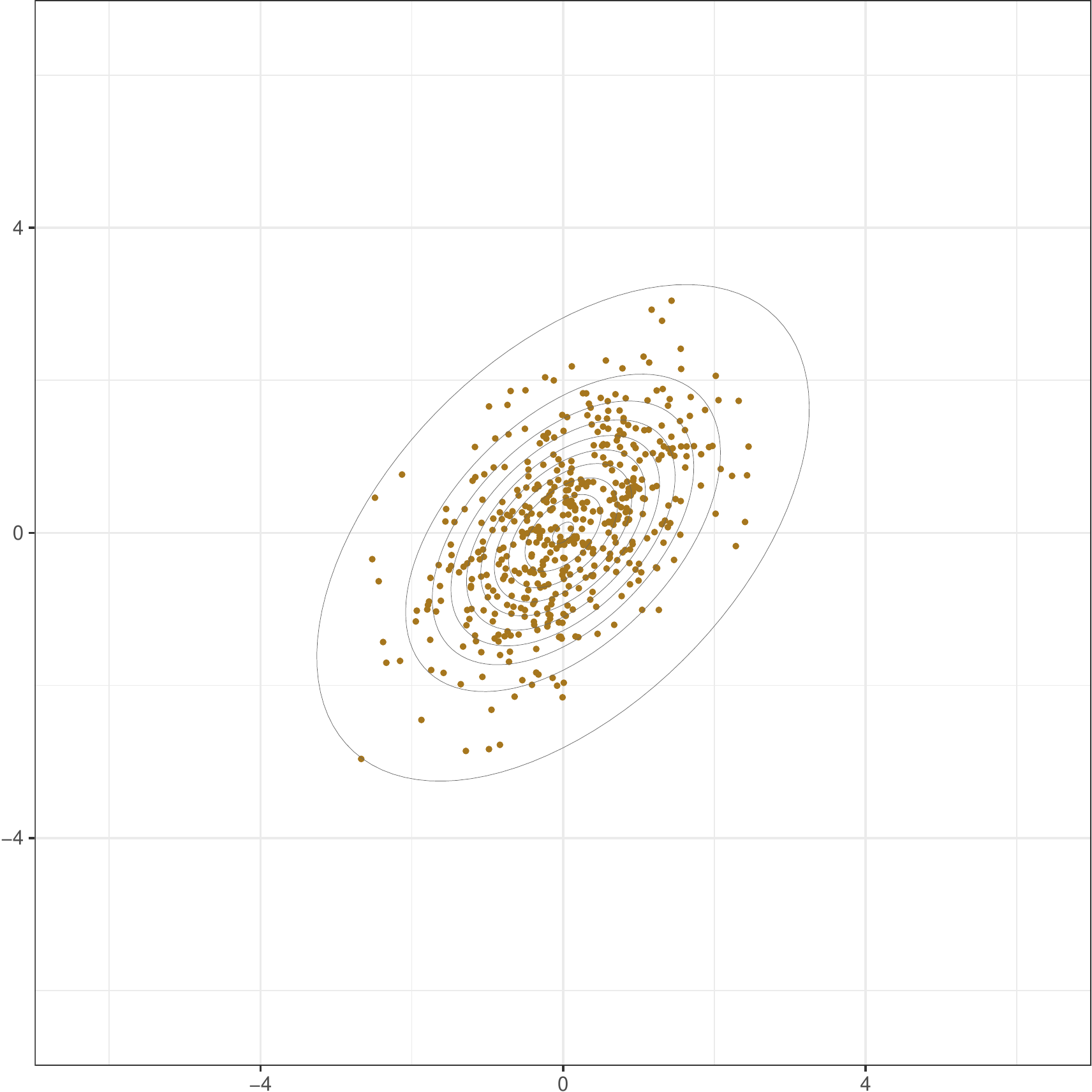}}}
\caption{Multivariate $t_\nu(\mu, \Sigma)$ densities and sample realizations for
  degrees of freedom $\nu=5$, 15,25 and $\infty$. Here $\mu = (0,0)$ while
  $\Sigma$ is the $2\times 2$ matrix with diagonal entries 1 and
  off-diagonal elements 0.5.} 
  \label{coutour.t}
\end{figure}
 The  $t_p(\mu, \Sigma; \nu)$ distribution has
 characteristic function as follows
 \begin{theorem}
 \label{t:cf}
Let $X$ be a $p$-dimensional random vector having the multivariate
$t$ distribution as per \eqref{mult:t}. Then the characteristic
function of $X$ is given by 
\begin{equation}
  \phi_X(t) =  \exp{(it'\mu)}\frac{2}{\Gamma(\frac{\nu}{2})}\Big(\frac{i\nu t'\Sigma t}{4} \Big)^{\frac{\nu}{4}}K_{\frac{\nu}{2}}(\sqrt{i\nu t'\Sigma t}),
  \label{eq:t:cf}
\end{equation}
where $K_{\frac{\nu}{2}}(s)$ is the modified Bessel function of the
second kind~\citep{abramowitzandstegun64} of order $\nu/2$ and $\Gamma (z) = \int_0^{\infty}\exp{\{-x\}} x^{z-1}dx$
is the gamma function.
\end{theorem}
\begin{proof}
From definitions of the characteristic function of $X$
and the expectation of functions of $X$, we have $\phi_X(it) =
E_X\{\exp{(it'X)}\} = E_{Y,S}[\exp{\{it'(\mu+Y\sqrt{\nu/ S})\}}] =
\exp{(it'\mu)}E_{Y,S}\{\exp{(it'Y\sqrt{\nu/ S})}\}.$ 
But from the definition of joint expectations,
$E_{Y,S}\{\exp{(it'Y\sqrt{\nu /S})}\} =  E_S\{E_{Y\mid 
  S}[\exp{\{(it'Y\sqrt{\nu/ S})\}}]\}$. Also $Y$ is independent of
$S$, so  the inner conditional expectation is given by $E_{Y\mid   
  S}[\exp{\{(it'Y\sqrt{\nu/ S})\}}] = \exp{\{{-i\nu t'\Sigma
    t}/({2 S})\}}$. We now use the result~\citep[see][pages 119 and 431]{bernardoandsmith93} that if
$S \sim \chi^2_\nu$, then $W=1/S$ has the inverse-$\chi_\nu^2$
distribution with characteristic function given by
\begin{equation}
\phi_W(r) = \frac{2}{\Gamma(\frac{\nu}{2})}\Big(-\frac{
ir}{2} \Big)^{\frac{\nu}{4}}K_{\frac{\nu}{2}}(\sqrt{-2ir}).
\end{equation}
Then, $E_S[\exp{\{{-i\nu t'\Sigma
    t}/({2S})\}}]$ is the same as evaluating $\phi_W(r)$ at $r=
\nu t'\Sigma t/2$. The theorem follows.
\end{proof}
\begin{corollary}
\label{corr:t}
Let $X$ be a $p$-dimensional random vector from the multivariate
$t$-density $t_p(\mu, \Sigma; \nu)$. Let $X^{(q)}$ be the first $q
(\leq p)$
coordinates in $X$. Then $X^{(q)}\sim t_p(\mu^{(q)}, \Sigma^{(q\times
  q)}; \nu)$, where $\mu^{(q)}$ is the vector with the first $q$
coordinates of $\mu$ and $\Sigma^{(q\times q)}$ is the 
$q\times q$ symmetric matrix containing the first $q$ rows and column
entries of $\Sigma$. 
\end{corollary}
\begin{proof}
Setting $t' = (t_1, t_2, \ldots, t_q, 0, 
0, \ldots,0)$ in \eqref{eq:t:cf} yields the characteristic function
that is uniquely that of the $t_{q}(\mu^{(q)}, \Sigma^{(q \times q)},
\nu)$ density and the result follows.
\end{proof}
 
\subsubsection{MBC with $t$-mixtures}
\label{sec:tmmbc}
MBC is an effective and principled method of obtaining  groups of
similar observations  in a dataset. It scores over partitioning
algorithms like $k$-means mainly in that it is not
restricted by the inherent assumption of homogeneous spherically dispersed groups. Assuming
spherically-dispersed groups when they are really
non-spherical can lead to erroneous results \citep[see][for a
comprehensive review on the pitfalls of using $k$-means when
assumptions are not met.]{chattopadhyayandmaitra17} In
MBC~\citep{melnykovandmaitra10,mclachlanandpeel00,fraleyandraftery02}
the  observations $X_1,X_2,\ldots,X_n$ are assumed to be
realizations from a \textit{K}-component  mixture
model~\citep{mclachlanandpeel00} with PDF 
 \begin{equation}
 	f(x;\theta) = \sum_{k=1}^{\textit{K}}\pi_k f_k(x; \eta_k)
    \label{mixedmodel}
 \end{equation}
 where $f_k(\cdot; \eta_k)$ is the density of the $k$th group,
 $\eta_k$ the vector of unknown parameters and $\pi_k  = Pr[x_i \in
 \mathit{G}_k]$ is the
 mixing proportion of the $k$th group, $k=1,2,\ldots,K$,  and
 $\sum_{k=1}^{\textit{K}}\pi_k  = 1$. For convenience, we write
 $\theta$ as the set of of all the model parameters. The most popular mixture model
 is the Gaussian Mixture Model (GMM), where the component  densities
 are assumed to be multivariate Gaussian 
 $\phi(x;\mu_k,\Sigma_k)$ with mean $\mu_k$ and dispersion
 $\Sigma_k$. Imposing different constraints on the densities (mostly
 on the dispersion matrices) gives rise to a family of mixture models
that are more parsimonious compared to the fully
 unconstrained model. The popular  
{\tt MCLUST} GMM family of~\citet{fraleyandraftery98,fraleyandraftery02} uses an
 eigen-decomposition of the component's variance covariance
 matrices. Thus, they write $\Sigma_k =
 \lambda_kB_k\Lambda_kB_k^T$, where $\Lambda_k$ is a diagonal matrix
 with values proportional to the eigen values of $\Sigma_k$, $B_k$
 denotes the matrix of eigen vectors of $\Sigma_k$ and $\lambda_k$ is
 the constant of proportionality. The {\tt MCLUST} family
 ~\citep{fraleyandraftery98,fraleyandraftery02} has 17 GMMs
 obtained by imposing certain constraints on
 $\lambda_k$, $B_k$ and $\Lambda_k$ and is implemented using the
 R~\citep{R} package {\tt MCLUST}~\citep{mclust1,mclust2}. Another
 useful class of 
 mixture models is obtained when the component cluster densities are assumed
 to follow a multivariate {\it t}-distribution rather than a
 multivariate Gaussian distribution. Motivated by
 \citet{mclachlanandpeel98}, these mixture models perform better than 
 GMM when it is plausible for each group has some extreme 
 observations.~\citet{jefferyandmcnicholas12} proposed a set of
 multivariate {\it t} mixture models({\it t}MM)  by imposing the same
 constraints as {\tt 
   MCLUST} plus additional constraints on $\nu$ to provide
  24 multivariate {\it t}MMs. These models are
  implemented in the {\tt TEIGEN}
 package~\citep{teigen2} in R~\citep{R}.  
 
The most common method of estimating the parameters of a mixture model
is the Expectation Maximization (EM)
Algorithm~\citep{dempsteretal77,mclachlanandkrishnan08}, which is an
iterative method for finding maximum likelihood estimates (MLEs) in
incomplete data scenarios. \citet{jefferyandmcnicholas12} used a variant
of EM known as the Expectation Conditional Maximization (ECM)
algorithm~\citep{mengandrubin93} to estimate the parameters of the
{\it tMM}. Faster
modifications~\citep{mengandvandyk97,chenandmaitra11} exist but they
have the same idea as ECM in that they replace the M-step of EM with
a sequence of  $D$ 
conditional maximization (CM) steps. Thus, the vector of parameters 
$\theta$ is partitioned into
$D$ sub-vectors $\theta = (\theta_1, \theta_2,\ldots, \theta_D)$ and 
the maximization done in $D$ steps, where the $d$th CM step maximizes
the Q function (or the expected complete log likelihood function given
the observations) over  $\theta_d$ but while keeping the other
sub-vectors fixed at some previous  
value.  ECM is  computationally more 
efficient than EM and  also  shares its desirable convergence
properties~\citep{mengandrubin93}. We now outline ECM estimation in
the context of $t$MMs.

Let $\zeta_{ik}$ be indicator variables that denote the cluster
membership of the $i$th observation. Thus
\begin{equation}
  \zeta_{ik}=\begin{cases}
    1, & \text{if the $i$th observation $x_i$ belongs to the $k$th group}\\
    0, & \text{otherwise}.
  \end{cases}
  \label{teig:miss}
\end{equation}
Note that $\zeta_{ik}$s  are unobserved and estimating them is the
major objective of MBC. In the context of the $t$MM, there is an
additional set of missing values
$v_{ik},i=1,2,\ldots,n;k=1,2,\ldots,K$
that are realizations from the Gamma density with PDF
\begin{equation}
\gamma(v_{ik};\nu_k /2,\nu_k /2) = \frac{\nu_k^{\frac{\nu_k}{2}} {v_{ik}}^{\frac{\nu_k}{2}-1} \exp{({-\frac{\nu_kv_{ik}}2})}}{2^{\frac{\nu_k}{2}}\Gamma (\frac{\nu_k}{2})}.
\label{den:gamma}
\end{equation}
Then for the  $t$MM  the complete data loglikelihood function can be written as
\begin{equation}
\begin{split}
\ell(&\pi, \mu, \Sigma, \zeta) \\
&= \sum_{k=1}^K \sum_{i=1}^n \zeta_{ik} \log\{\pi_k \phi(x_i;\mu_k,
\Sigma_k/v_{ik})\gamma(v_{ik};\nu_k /2,\nu_k /2)\} \\
\end{split}
\label{teig:likelihood}
\end{equation}
where $\nu_k$ denotes the degrees of freedom of the $t$-density for the $k$th group, $\phi$ denotes Gaussian density with mean $\mu_k$ and variance $\Sigma_k /v_{ik}$.
The $t$MM likelihood is maximized through the following steps:
\begin{enumerate}
\item {\em Initialization}. Let $\{(\pi_k^\circ, \mu_k^\circ,
  \Sigma_k^\circ, \nu_k^\circ); k=1, 2,\ldots,K\}$ be the initializing
  parameters. 
\item {\em E-step updates}. The component weights $v_{ik}$ and the
  group indicator variables $\zeta_{ik}$ are updated as
\begin{eqnarray*}
\hat{v}_{ik}&=&\frac{\nu_k^\circ + p}{\nu_k^\circ+ (x_i-\mu_k^\circ)^T\Sigma_k^{\circ -1}(x_i-\mu_k^\circ)}\\
\hat{\zeta}_{ik}&=&\frac{\pi_k^\circ f_t (x_i;\mu_k^\circ, \Sigma_k^\circ, \nu_k^\circ)}{\sum_{k=1}^K\pi_k^\circ f_t (x_i;\mu_k^\circ, \Sigma_k^\circ, \nu_k^\circ)}
\end{eqnarray*}
where $f_t (x_i;\mu_k^\circ, \Sigma_k^\circ, \nu_k^\circ)$ denotes a multivariate {\it t}-density with mean $\mu_k$, dispersion matrix $\Sigma_k$ and degrees of freedom $\nu_k$.
\item{\em CM-step 1}. The first CM step updates the component means
  $\mu_k$s and the prior probabilities $\pi_k$s:
\begin{eqnarray*}
\hat{\pi}_k&=&\frac{\hat n_k}{n}\\
\hat{\mu}_k&=&\frac{\sum_{i=1}^n\hat{\zeta}_{ik}\hat{v}_{ik}x_i}{\sum_{i=1}^n\hat{\zeta}_{ik}\hat{v}_{ik}}\\
\end{eqnarray*}
where $\hat n_k =\sum_{i=1}^n\hat{\zeta}_{ik}$. Additionally, if  the
constraint $\nu_k \equiv \nu$ is imposed upon the degrees of freedom of each
group, then $\hat\nu$ is updated here by solving the equation
\begin{equation}
\begin{split}
1-\psi\Big(\frac{\hat{\nu}}{2}\Big)+\frac{1}{n}&\sum_{k=1}^K\sum_{i=1}^n\hat{\zeta}_{ik}(\log\hat{v}_{ik}-\hat{v}_{ik})+\log\Big(\frac{\nu}{2}\Big)\\&+\psi\Big(\frac{\nu^\circ+p}{2}\Big)-
\log\Big(\frac{\nu^\circ+p}{2}\Big)=0.
\end{split}
\end{equation}
where $\hat{\nu}$ denotes the updated degrees of freedom and $\nu^\circ$ denotes the current estimate.
\item{\em CM-step 2}. This step updates the $\Sigma_k$s and varies
  accordingly as
  per constraints imposed in the modeling. For example, setting
  $\Lambda_k=\Lambda$ and $B_k=B$ yields the updates 
\begin{equation}
\hat{\lambda}_k=\frac{1}{p\hat n_k} trace\Big(S_k H^{-1}\Big)
\end{equation}
where $S_k=\frac{1}{\hat
  n_k}\sum_{i=1}^n\hat{\zeta}_{ik}\hat{v}_{ik}{||x_i-\hat{\mu}_k||}^2$
and $H=B \Lambda B^T$.  In order to update $B$ and $\Lambda$, $H$ is
first updated using using 
\begin{equation}
H=\frac{\frac{1}{\lambda_k}\sum_{k=1}^K S_k}{{|\frac{1}{\lambda_k}\sum_{k=1}^K S_k|}^\frac{1}{p}}
\label{ecm:D}
\end{equation}
Now the updated $B$ and $\Lambda$ are obtained from~\eqref{ecm:D}.
\item Alternate between the E- and CM-steps till convergence.
\end{enumerate}
After obtaining the final estimates of the parameters, the $i$th data
point $X_i$ is assigned to the class for which the converged E-step
posterior probability is the highest, that is $X_i$ is assigned to class
$k$ where $k =\argmax_{l}\hat{\zeta}_{il}$.

Our $t$-mixtures MBC ($t$MMBC) formulation above assumes a known number
of components $K$. With unknown $K$, a  popular approach of finding it
is the Bayesian Information Criterion (BIC)~\citep{schwarz78}
that subtracts  $(m\log n)/2$ from the maximized log-likelihood
(obtained from the converged ECM), with $m$ the
number of unconstrained parameters in the fitted $K$-component
$t$MM. ~\citep[See][for a detailed review of
BIC.]{chattopadhyayandmaitra17} 

\subsubsection{Illustration}
\label{illustrate:t-mix}
The purpose of this section is to demonstrate using a simulated
dataset, potential pitfalls in fitting a GMM on data that are plausibly
realizations from a multivariate {\it tMM}. For this purpose, we
simulated three 4-component multivariate {\it t}MMs having an
approximate generalized overlap~\citep{maitra10,melnykovandmaitra11}
of $\ddot{\omega} =0.01$ and $\nu=5, 15, 25$ degrees of freedom. The
generalized overlap measure 
is an effective way of summarizing the overlap matrix $\Omega$ of
\citet{maitraandmelnykov10}  which
is a $K \times K$ matrix whose $(i, j)th$ element contains the sum of
misclassification probabilities between $i$th and $j$th clusters. The
generalized overlap summarizes this matrix and is 
defined as $\ddot{\omega} = (\lambda_{(1)}-1)/(K-1)$ where
$\lambda_{(1)}$ is the largest eigen value of $\Omega$. Small values
of $\ddot{\omega}$ are likely to indicate more distinct
groupings.  \citep[See also][for further
details.]{chattopadhyayandmaitra17}.
 The {\tt MixGOM()} function in the R package {\sc MixSim}
 \citep{melnykovetal12} can be adapted in the same 
manner as \citet{melnykovandmaitra11} to obtain realizations from a
multivariate $t$MM with a specified $\ddot\omega$.

We fitted both {\it GMM} and {\it tMM} to the datasets, with  optimal
number of groups determined using BIC, and results displayed in Fig.
\ref{sim.gauss.t}. For $\nu=5$ and $\nu=10$ 
fitting a {\tt GMM} gives the optimal number of clusters to be five and
four respectively but for {\it tMM} the optimum number of clusters were
correctly chosen as three in both cases. For $\nu=25$, BIC identified the
number of clusters as three for both class of models. 
\begin{figure}
\includegraphics[width=0.5\textwidth]{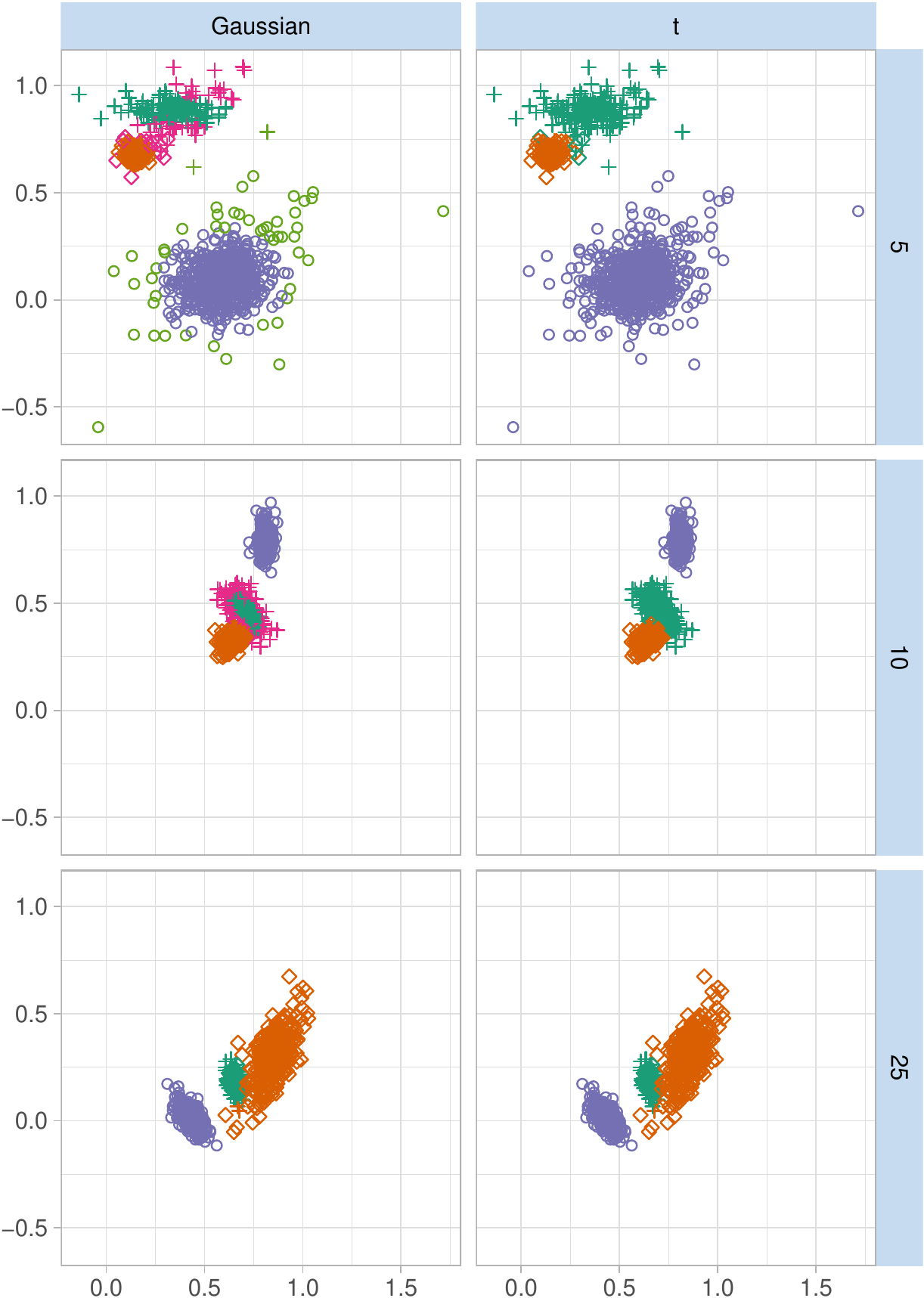}
\caption{Simulated datasets from a bivariate {\it t}MM with $\nu =5,
  10, 25$ in the top, middle and bottom panels, and clustering
  obtained using GMM (left-hand panel) and $t$MM (right-hand panel). In all
  figures, plotting character indicates true classification while
  color indicates estimated grouping.}
\label{sim.gauss.t}
\end{figure}
For numerical assessment of clustering performance,we computed the
Adjusted Rand index ($\mathcal R$) after fitting each of the 
two models to the data. As also explained by \citet{maitra01}, the Rand
index~\citep{rand71} is a measure of similarity between two different
clusterings and is calculated as follows. Suppose that for a
given data set $\mathcal{D}$ having $n$ elements there exists two
different partitions $P_1$ and $P_2$ by two different algorithms. Let
$a_1$ denote the number of pairs of object in $\mathcal{D}$ that are
in the same group in both partitions $P_1$ and $P_2$, and $a_2$ the
number of pairs of objects in $\mathcal{D}$ that are in different
groups. Then the Rand index is
defined as $(a_1 + a_2)/\binom{n}{2}$. The Rand Index takes values
between 0 and 1. The Adjusted Rand Index~\citep{hubertandarabie85} 
 is obtained by correcting
Rand Index for chance grouping of elements and can take values between
$-\infty$ and $1$.  We analyzed the quality of clustering for
both GMM and $t$MM by comparing the clustering results with the true
class indicators through $\mathcal R$ in order to
demonstrate issues of using a GMM model for clustering a data
supposedly originating from a {\it 
  t}MM, that is from a mixture model with potentially heavier tails.
The results are presented in Table \ref{tab:sim:data}.    
\begin{table}
\caption{Adjusted Rand Indices obtained by clustering three datasets simulated from a 3-component tMM with $\nu$ equal to  5, 10 and 25 using a Gaussian mixture model and a $t$ mixture model.}
\centering
\begin{tabular}{P{2.5cm} P{2.5cm} P{2.5cm}}
\hline
$\nu$ & Gaussian & $t$\\
\hline
5 & 0.84 & 0.99\\
10 & 0.79 & 0.99 \\
25 & 0.99 & 0.99  \\
\hline
\end{tabular}
\label{tab:sim:data}
\end{table}

Both Fig.~\ref{sim.gauss.t} and Table~\ref{tab:sim:data}  clearly
indicate that a {\it t}MM gives a better fit that a {\tt GMM}, for smaller $\nu$. For $\nu = 25$, 
the difference is negligible but enough to prove that {\it tMM}
wins over {\tt GMM}. Our examples here illustrate the value  of
considering a {\it t}MM when the underlying groups 
are potentially thicker-tailed.
\subsection{Variable Selection in Clustering}
\begin{figure*} 
  \begin{minipage}[c]{0.47\textwidth}
    \subfloat[Pairwise scatterplots, estimated marginal densities and
    pairwise correlations of simulated datasset.]{\label{fig:dim:scatter}\includegraphics[width=\textwidth]{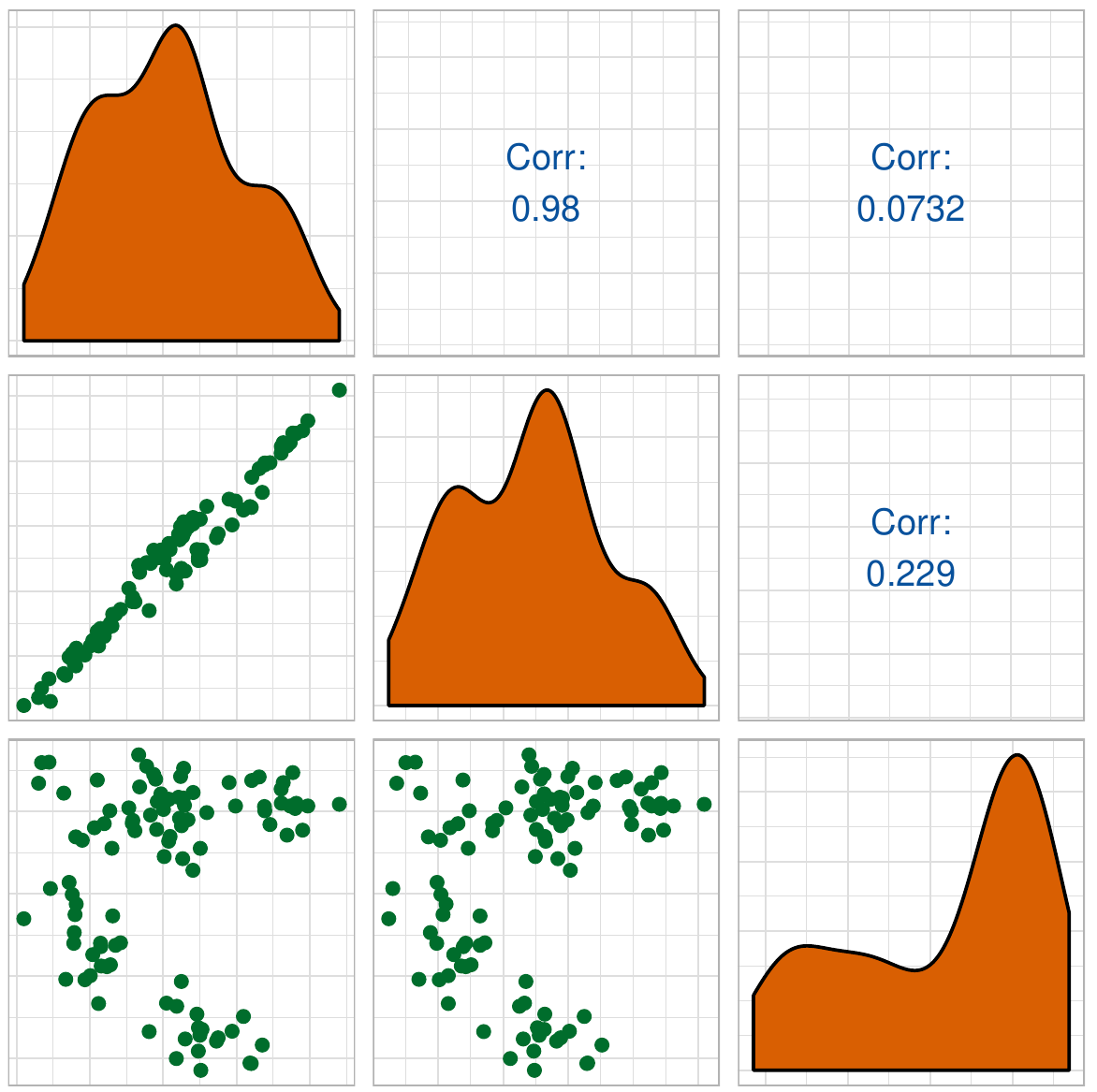}}
\end{minipage}\hspace{-0.01\textwidth}
\begin{minipage}[c]{0.53\textwidth}
\mbox{ 
\subfloat[$\hat K=2$, $\mR=0.37$]{\label{fig:dim:1:2}\includegraphics[width=0.333\textwidth]{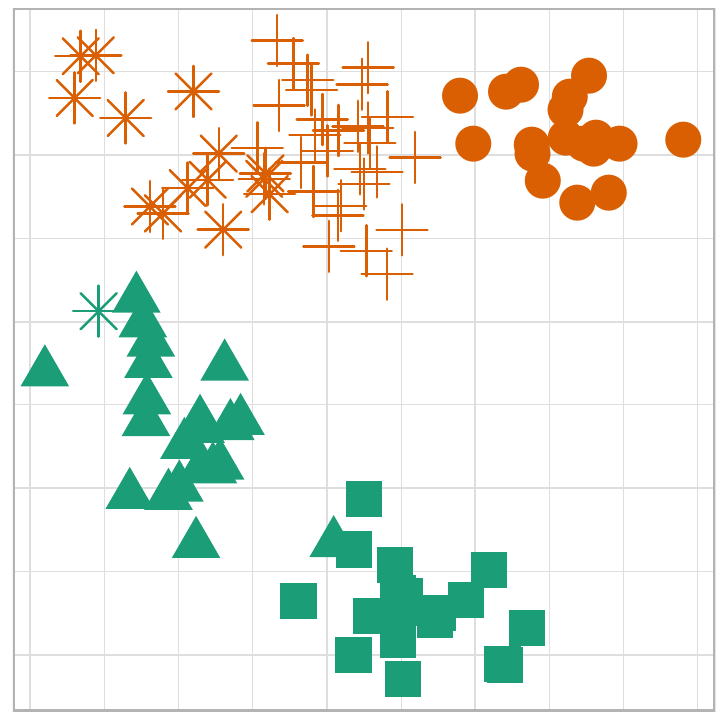}}
\subfloat[$\hat K=2$, $\mR=0.37$]{\label{fig:dim:1:3}\includegraphics[width=0.333\textwidth]{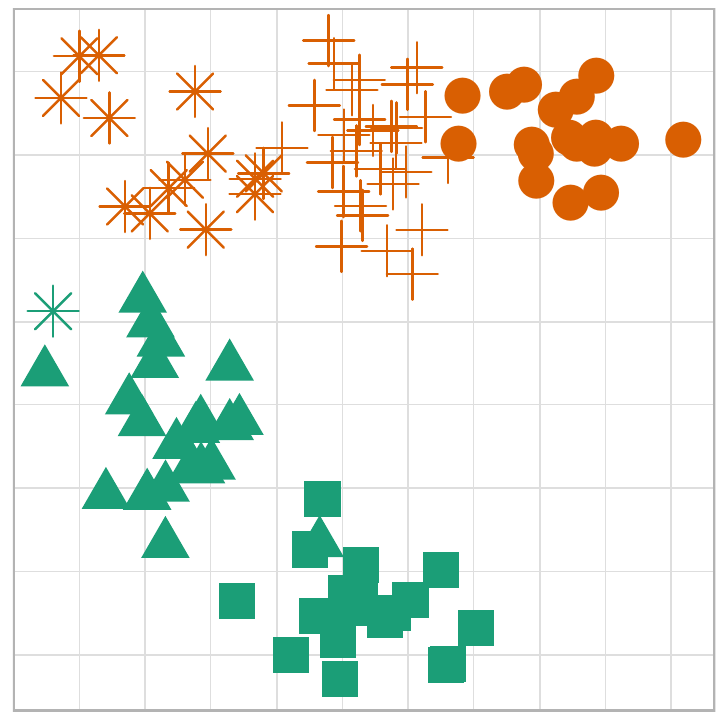}}
\subfloat[$\hat K=4$, $\mR=0.82$]{\label{fig:dim:2:3}\includegraphics[width=0.333\textwidth]{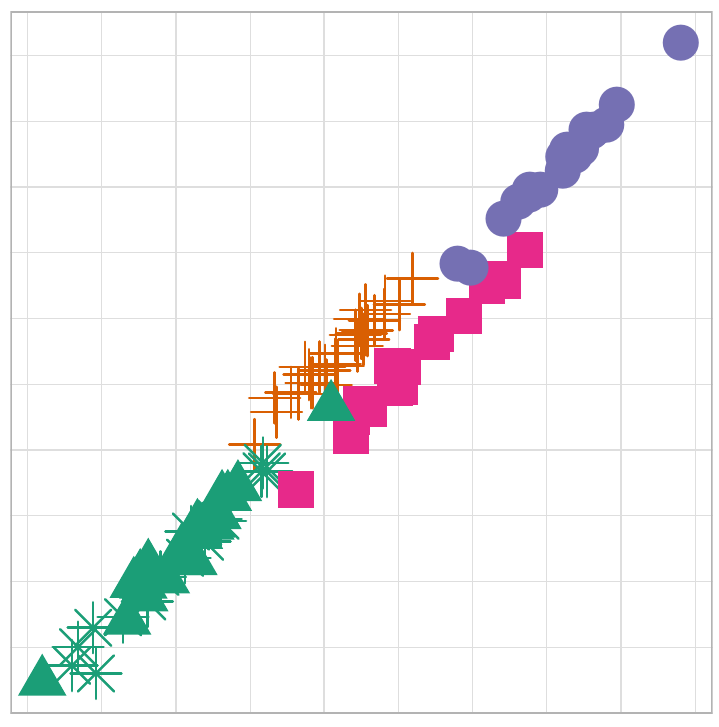}}
}
\centering
\mbox{
  \subfloat[Two views of clustering using all three dimensions: $\hat K = 5$, $\mR = 1.0$]{\label{fig:dim:all}\includegraphics[width=0.57\textwidth]{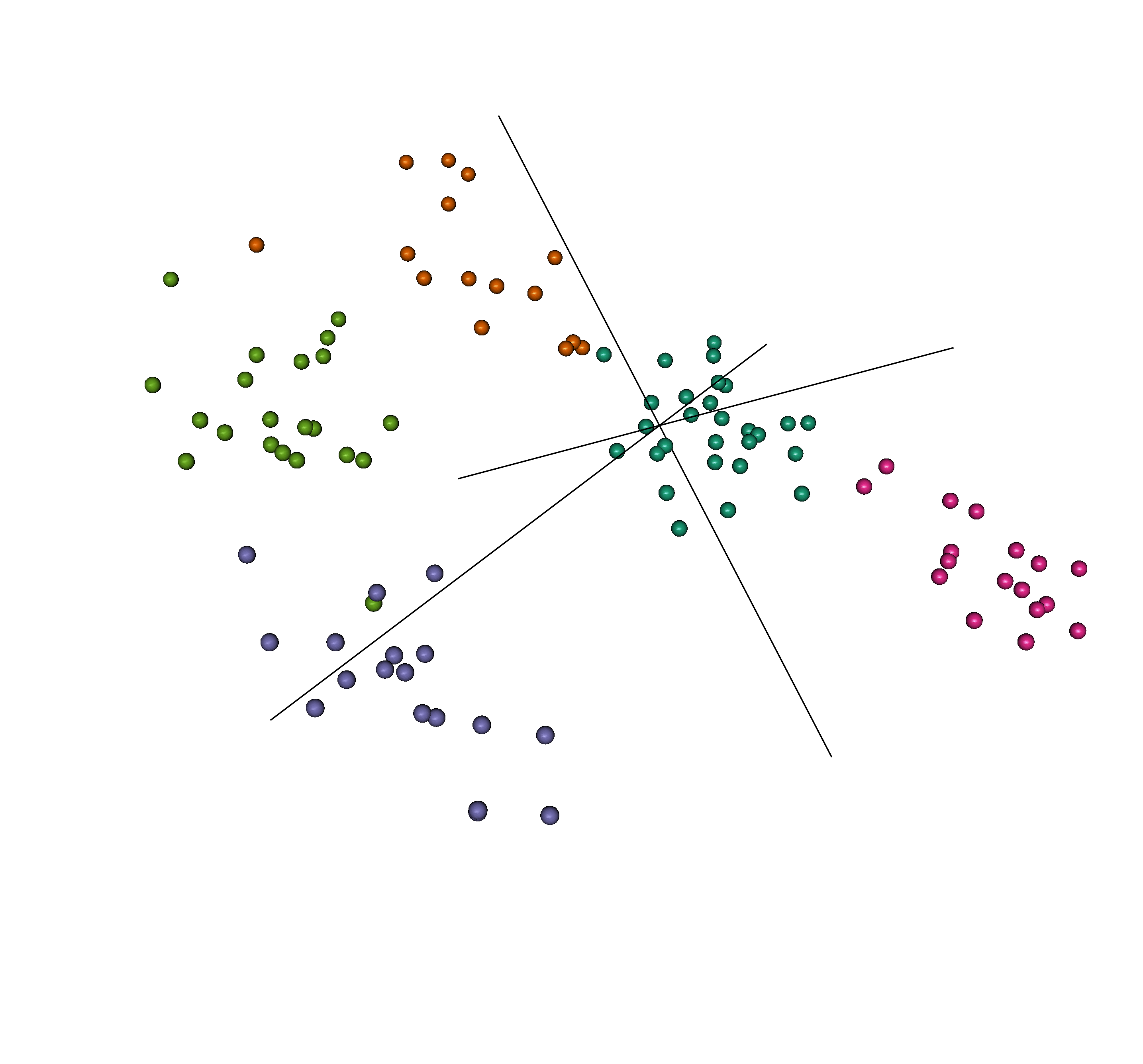}
    \hspace{-0.14\textwidth}\includegraphics[width=0.57\textwidth]{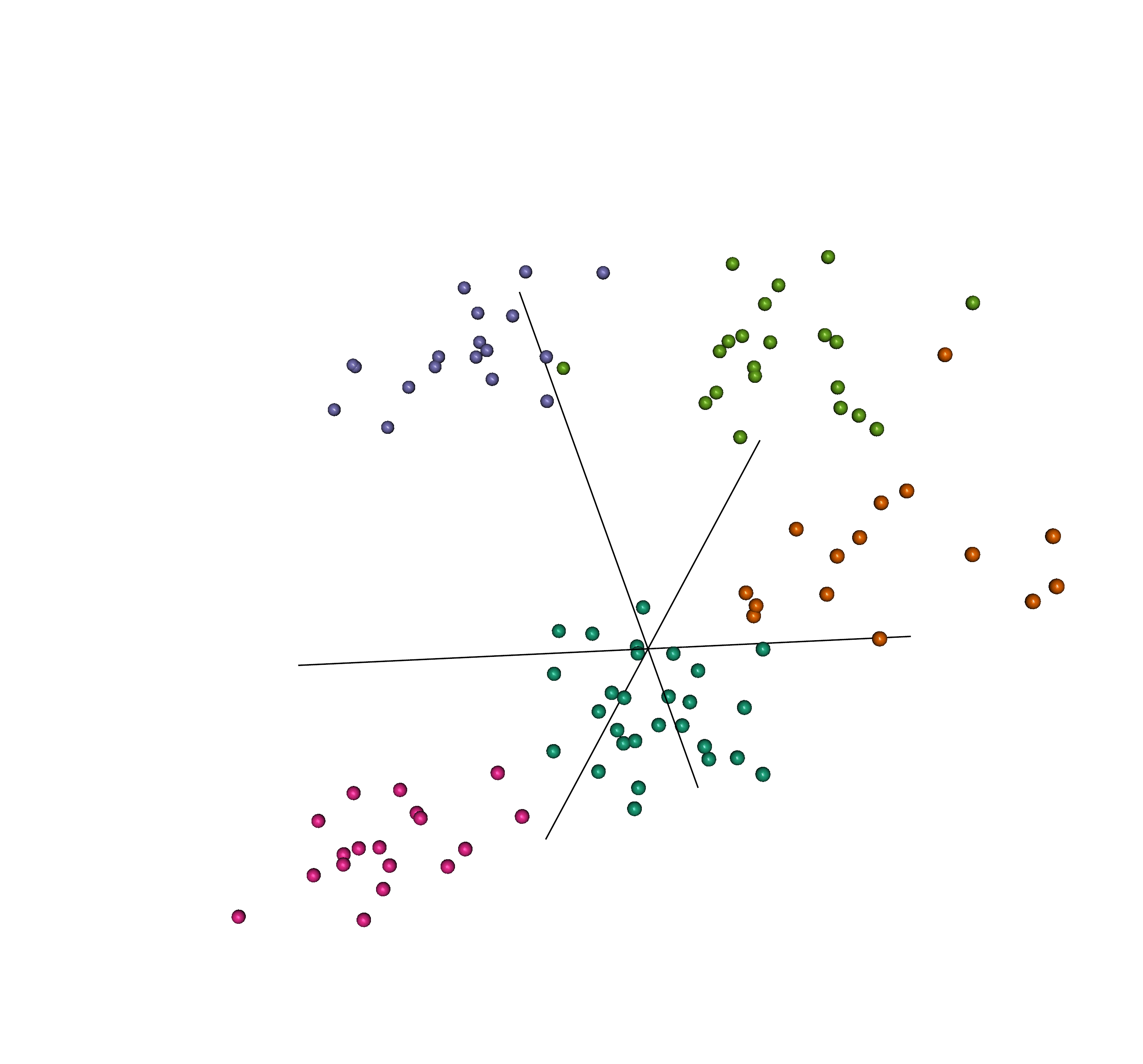}}
}
\end{minipage}
\caption{Result of clustering (a) data simulated  from a three-dimensional
  5-component {\tt GMM} using (b, c, d) two and (e) all three variables.}
\label{sim-norm-var-select}
\end{figure*}
\label{ov:variable selection}
Selection of relevant variables is a very important issue in
clustering. Incorporating redundant information can degrade overall
clustering performance producing less distinct
groups~\citep{chattopadhyayandmaitra17}. In the same way, exclusion of
variables having relevant 
information also degrades the overall quality of clustering. We
illustrate this problem (Fig.~\ref{sim-norm-var-select}) by means of a simulated three-dimensional
dataset drawn from a {\it GMM} with $K=5$ true and very well separated
components ($\ddot\omega = 10^{-5}$) 
and with marginal distributions and pairwise scatterplots as shown in the
diagonal and lower triangle of the matrix of plots in Fig.
\ref{fig:dim:scatter}. The upper triangle of these matrix displays
the correlation between each pairs of variables. 
Thus the first two of the three variables have a very high (almost linear)
pairwise correlation of 0.98, while the other two pairs have little to
modest correlations between them. It is tempting to surmise that one
of the first two dimensions are nearly redundant and dropping one of
them would not have much of an effect on the quality of
clustering. We test this assertion by performing GMMBC (with BIC used
to determine the optimal number of groups) on the dataset
using each of the three distinct pairs of
variables. Figs~\ref{fig:dim:1:2}-\ref{fig:dim:2:3} display the results of clustering each pair of
coordinates, with colors representing the obtained clustering and
character the true grouping. Including either one of the two
seemingly redundant variables with the third variable
identifies only two groups and a partitioning (Figs~\ref{fig:dim:1:2}
and \ref{fig:dim:1:3}) that hardly matches the true with an adjusted
Rand index ($\mR$) of 0.37. Interestingly, clustering the two
seemingly redundant variables (and ignoring the third) does better, identifying
four groups and with $\mR=0.82$. However, this is still a far cry from
the perfect partitioning that is obtained when all three coordinates
are used in GMMBC with BIC to determine the optimal number of groups.
Therefore, it is important to consider the relevant variables in
clustering. This fact becomes more important in the light of attempts
by many researchers to cluster GRBs using, for instance, only the
duration variables. This example also demonstrates that two variables
having correlations even as high as 0.98 are not
necessarily redundant and should be carefully analyzed before arriving at decisions 
on their inclusion or exclusion in analysis. 
\citet{rafteryanddean06} proposed a method for selecting variables
containing the most relevant clustering information by recasting the
variable selection problem in terms of model selection, where
comparison of the models is done via
BIC. \citep[refer to][for a thorough review on
variable selection.]{chattopadhyayandmaitra17}

\subsection{Measuring Distinctiveness of Partitioning Through the  Overlap}
\citet{chattopadhyayandmaitra17} explain that an overlap measure can
be used to indicate the extent to which  clusters obtained through a method are
  distinct from one another. We refer to that paper for further details,
  but note that they adopt \citet{maitraandmelnykov10}'s definition of
  the pairwise overlap between two groups as the sum of their
  misclassification probabilities. To provide a sense of these
  pairwise overlap measures, we illustrate  three two-dimensional examples in
  Fig.~\ref{overlap-contour}. (For ease of display and
  understanding, we use two dimensions here, but the general idea is
  the same for all dimensions.) In each case, we used the {\tt MixSim}
  package to sample 1000
  observations from two-component GMMs but with pairwise overlaps of
  $\omega=10^{-3}$, 0.05 and 0.1, respectively. 
\begin{figure}
\mbox{ 
\subfloat[$\omega = 10^{-3}$]{\label{ov:contour:0}\includegraphics[width=0.16\textwidth]{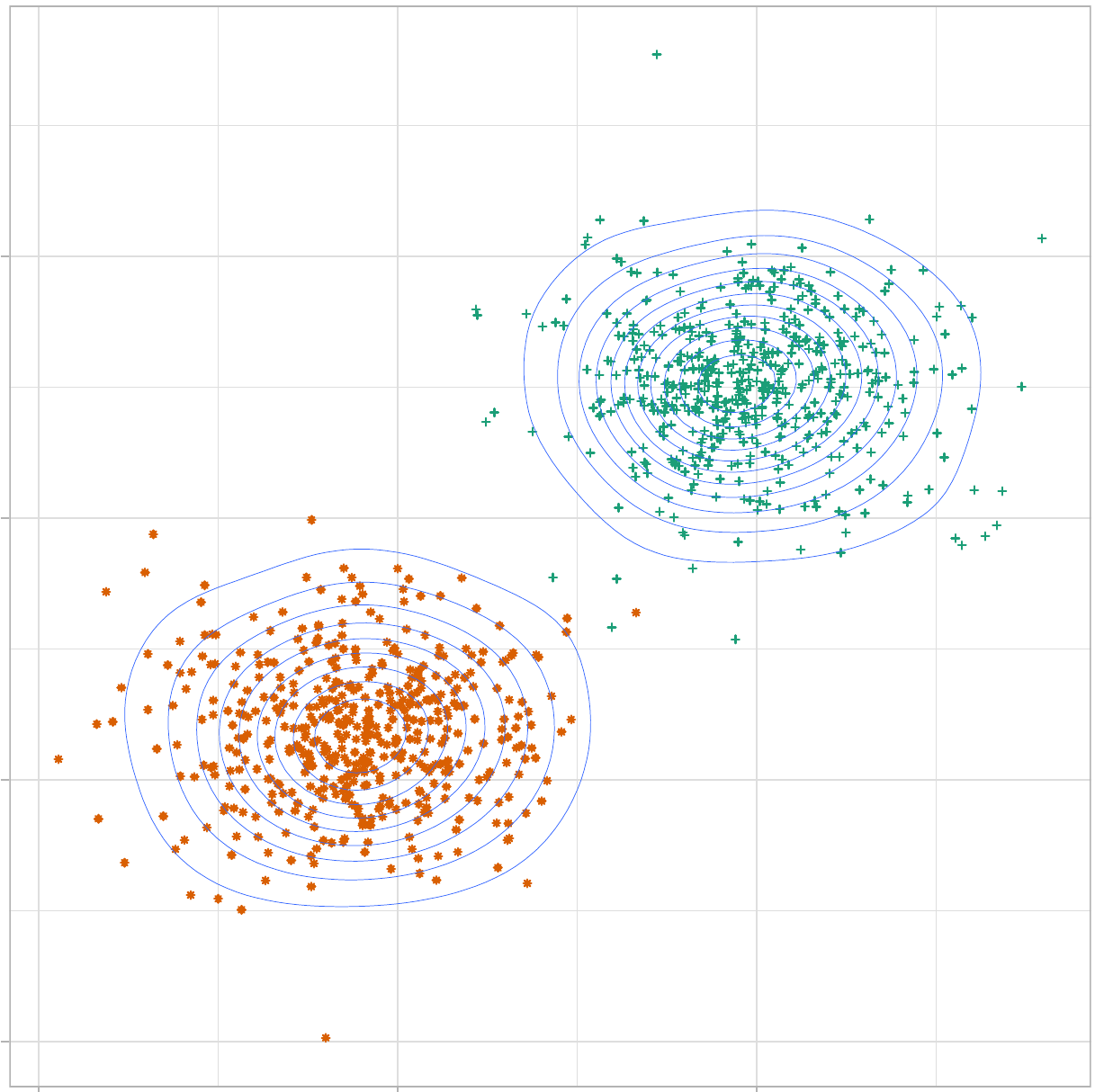}}
\subfloat[$\omega = 0.05$]{\label{ov:contour:0.04}\includegraphics[width=0.16\textwidth]{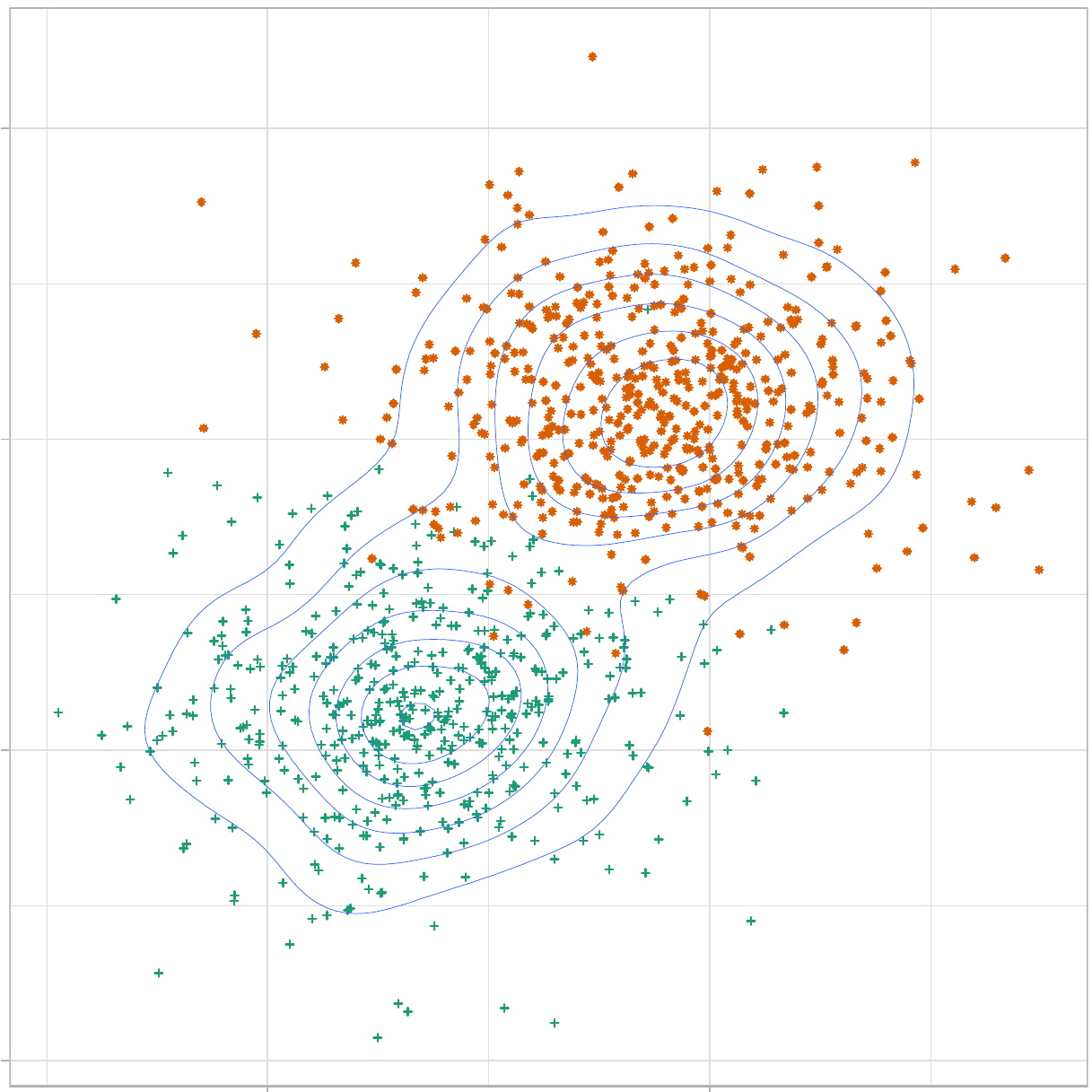}}
\subfloat[$\omega = 0.1$]{\label{ov:contour:0.1}\includegraphics[width=0.16\textwidth]{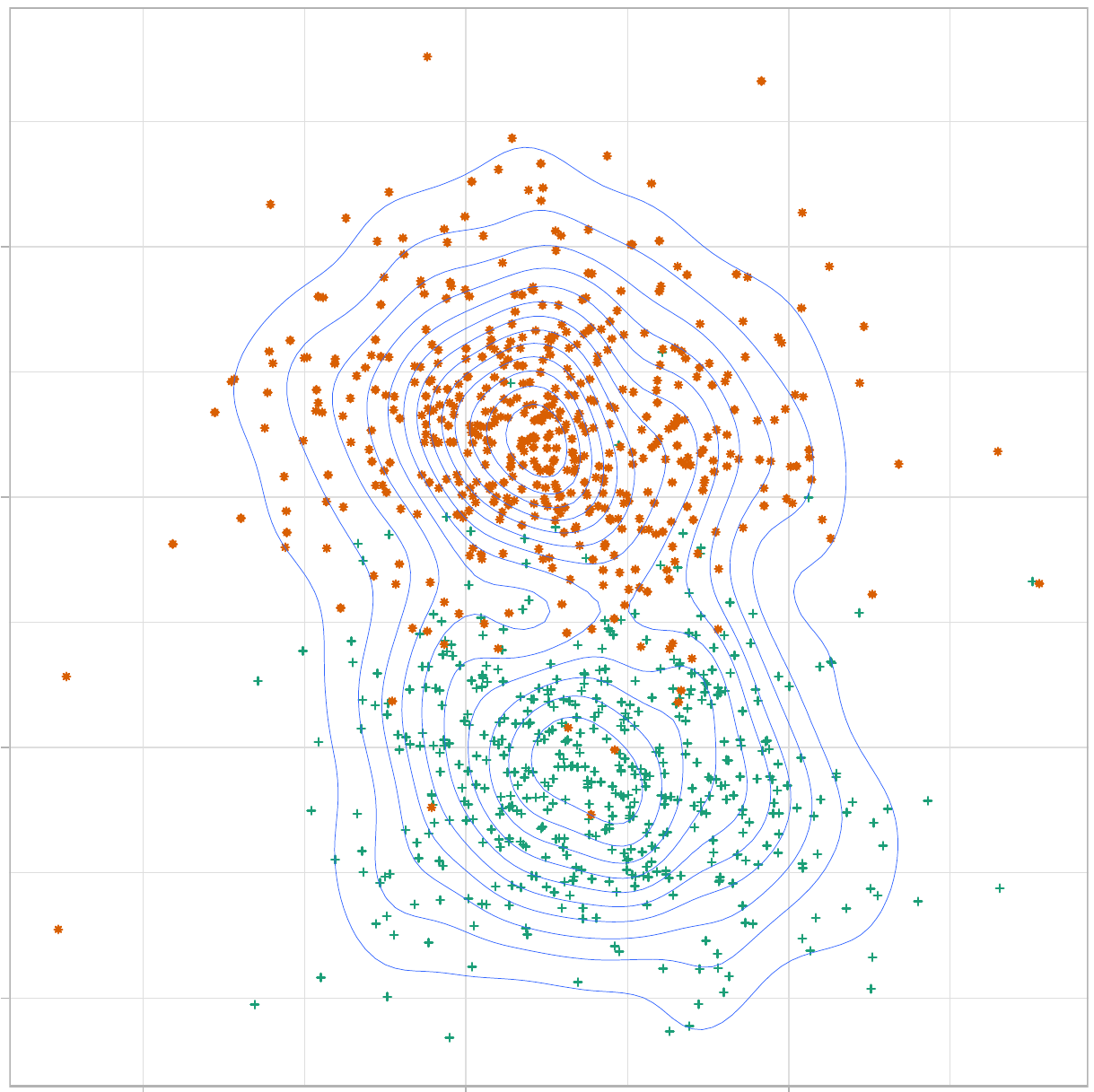}}}
\caption{Sample realizations, along with estimated contour densities,
  from two-component  $t$MMs with three   different overlap values.}
  \label{overlap-contour}
\end{figure}
In each figure, we display the observations from each group by means
of color and character and also provide an estimated bivariate density
through a contour plot. The contours are totally separated
for the data with negligible pairwise overlap
(Fig.~\ref{ov:contour:0}). This separation decreases with increasing
overlap (Figs  \ref{ov:contour:0.04} and \ref{ov:contour:0.1}) and
the groups become less distinct. 
\subsection{Classification}
\label{ov:classification}
The objective of classification is to classify a new observation to
one of $K$ well-defined groups. Classification is a supervised
learning method that uses training data to determine a rule that 
classifies a new observation to one of $K$ groups.  Bayes' rule is
often used to obtain classification rules in the model-based
context. Here an observation $x$ is assigned to  the $l$th group if
the posterior probability of $x$ belonging to the $i$th group is
the highest amongst all groups under considerations. Thus, the decision rule
is to classify $x$ to group $i$ if $\pi_i f_i(x) > \pi_k f_k(x)$ for
all $k\neq i$, where $\pi_j$ is the prior probability that an observation
belongs to the $j$th group and $f_j(x)$ denotes the PDF of the $j$th group at $x$ \citep[see][for more details.]{johnsonandwichern88}. In case of tMM $f_i(x)$ denotes a
multivariate-$t$ PDF. 
Our proposal is to cluster the observations for which all parameters
are observed and to use the estimated parameters to classify the GRBs
for which not all parameters are observed. We will use the above
methods for developing our classification rule for observations with
missing records. 
\section{Cluster Analysis of GRBs}
\label{GRB:1599:full}
The BATSE catalogue is widely used for analysis
of GRBs and has temporal and spectral information of
GRBs from 1991 to 2000. A few of the parameters have been
of interest to researchers for grouping GRBs. These are:
\begin{description}
\item[$T_{50}$:]{the time by which 50\% of the flux arrive.}
\item[$T_{90}$:]{the time by which 90\% of the flux arrive.}
\item[$P_{64}$, $P_{256}$, $P_{1024}$:]{the peak fluxes measured in
  bins of 64, 256 and 1024 ms, respectively.} 
\item[$F_{1}$, $F_{2}$, $F_{3}$, $F_{4}$:]{the four time-integrated
  fluences in the  20-50, 50-100, 100-300, and $>$ 300 keV spectral channels,
  respectively.} 
\end{description}
Apart from these nine parameters three more composite parameters  are
of interest to researchers~\citep{mukherjeeetal98}. These are: 
\begin{description}
  \item [$F_t$=$F_1+F_2+F_3+F_4$:]{the total fluence of a
GRB.}
  \item [$H_{32}= F_3/F_2$:]{measure of spectral hardness using the
    ratio of $F_2$ and $F_3$.}
  \item [$H_{321}= F_3/(F_1 + F_2)$:]{measure of spectral hardness
    based on the ratio of channel fluences $F_1,F_2,F_3$.}
\end{description}
The current (BATSE 4Br) catalogue  contains revised
locations of 208 bursts from the BATSE 4B Catalog along with 515
bursts observed between September 20 1994 and August 29 1996 apart
from the bursts present in the BATSE 3B Catalog. Many parameters,
largely  in the four time-integrated fluences $F_1$, $F_2$, $F_3$
and especially  $F_4$, have zeroes recorded that can be regarded as
missing \citep{chattopadhyayandmaitra17}. Consequently, the 
derived variables are also missing for these GRBs. So
the BATSE 4Br catalog has 1599 GRBs (among 1973 GRBs)
containing complete information on all the nine original (plus
three derived) variables. Most authors have used a very subset of
these 12 variables for their analysis. This led us to wonder whether
the nine original variables contains relevant information that might
improve the quality of the clustering to give more coherent groups. We
thus analyzed 1599 GRBs from the BATSE 4Br catalogue having complete
information on the nine original variables with MBC using mixtures of
{\it t}-densities. Further, we wondered if the five groups
identified  by \citet{chattopadhyayandmaitra17} were only
because the groups identified were constrained to be Gaussian and so had thinner tails,
while there really were fewer groups with thicker tails, that is a
scenario reminiscent of the situation in
~\ref{sim.gauss.t}. Therefore, we reanalyzed the BATSE data using
all nine parameters and $t$MMBC.

We first briefly discuss the univariate and bivariate relationships
between the nine original parameters in the BATSE 4Br catalogue. 
\begin{figure*}
\includegraphics[width=\textwidth]{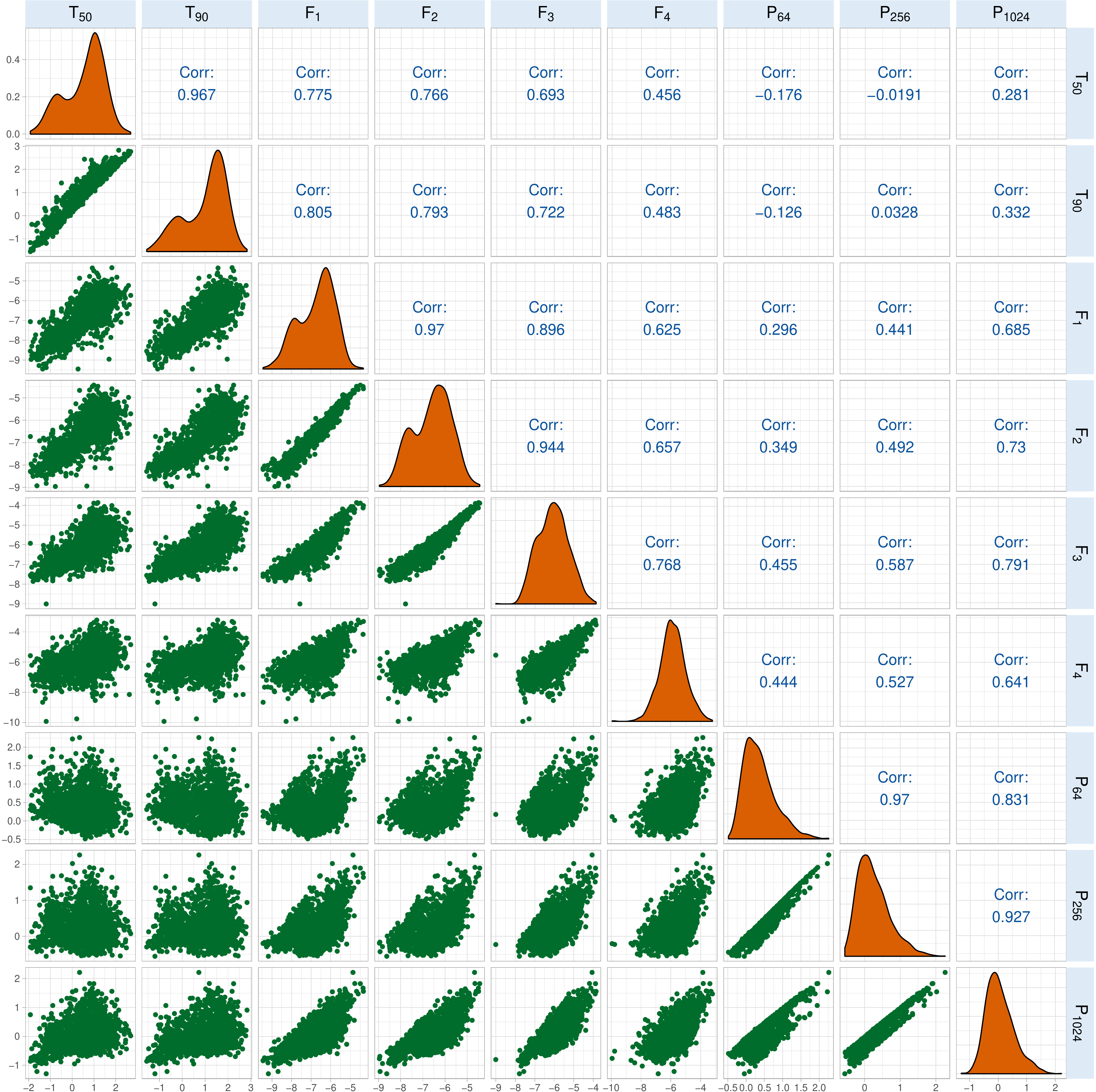}
\caption{A matrix of scatterplots (the lower triangle),
  density plots ( the diagonal) and correlation coefficients(the upper
  triangle) of the nine parameters $T_{50}$, $T_{90}$, $P_{64}$
  $P_{256}$, $P_{1024}$, $F_1$, $F_2$, $F_3$ and $F_4$ using 1599 GRBs of the BATSE
  4Br catalogue. All displays are  in the logarithmic scale.}
\label{fig:Bivariate9}
\end{figure*}
Fig.~\ref{fig:Bivariate9} displays the bivariate relationships
between the nine original variables of the BATSE 4Br catalogue along
with the univariate density plots of the nine parameters. The two
duration variables $\log_{10}T_{50}$ and $\log_{10}T_{90}$ exhibit a
very high positive association amongst themselves. Similar behavior is
exhibited by the three peak fluxes $\log_{10}P_{64}$,
$\log_{10}P_{256}$ and $\log_{10}P_{1024}$. Fluence $\log_{10}F_1$
shows a very high positive association with fluences $\log_{10}F_2$
and $\log_{10}F_3$ and a high positive association with
$\log_{10}F_4$. Fluence $\log_{10}F_2$ exhibits a very high positive
association with $\log_{10}F_3$ and a high positive association with
$\log_{10}F_4$ that also has a high positive association with
$\log_{10}F_3$. Duration $\log_{10}T_{50}$ exhibits a high positive
association with fluences $\log_{10}F_1$, $\log_{10}F_2$ and
$\log_{10}F_3$ and a moderate positive association with
$\log_{10}F_4$. $\log_{10}T_{90}$ behaves similar to $\log_{10}T_{50}$
except that it shows a very high positive association with
$\log_{10}F_1$. Fig.~\ref{sim-norm-var-select} has also pointed out
through the scatterplots the limitations that are posed on grouping
using only one or two variables, thus pointing out the importance of
using more than two variables for clustering. We now perform
cluster analysis on the $1599$ GRBs using the nine original parameters
(in logarithmic scale), that is $\log_{10}T_{50}$, $\log_{10}T_{90}$,
$\log_{10}F_1$, $\log_{10}F_2$, $\log_{10}F_3$, $\log_{10}F_4$,
$\log_{10}P_{64}$, $\log_{10}P_{256}$ and $\log_{10}P_{1024}$.  

\subsection{Clustering  GRBs Using All Observed Parameters}
\label{MBC:all}
We first perform $t$MMBC using 1599 GRBs from the BATSE 4Br catalogue
and then classify the GRBs with incomplete information to the groups
obtained using the $t$MMBC. 
\subsubsection{$t$MMBC with all nine parameters}
\label{t:MBC}
We check for redundancy among
the nine parameters $\log_{10}T_{50}$, $\log_{10}T_{90}$, $\log_{10}P_{64}$,
$\log_{10}P_{256}$, $\log_{10}P_{1024}$, $\log_{10}F_1$,
$\log_{10}F_2$, $\log_{10}F_3$, $\log_{10}F_4$ using model-based
variable selection. The results obtained (Table  \ref{tab:var:select}) do not show redundancy among these nine parameters. However, there is
redundancy beyond these nine variables, because the derived variables
are linearly related to the nine parameters.  We thus performed $t$MMBC
on the nine original variables using the {\tt TEIGEN}
package in R and determined $K$ from amongst $\{1, 2, \ldots,9\}$ using
BIC -- indeed, 
\begin{table}
\caption{Results of forward and backward-variable selection for determining redundancy among  $\log_{10}T_{90}$, $\log_{10}T_{50}$,
  $\log_{10}P_{64}$, $\log_{10}P_{1024}$, $\log_{10}P_{256}$,
  $\log_{10}F_{3}$, $\log_{10}F_{2}$, $\log_{10}F_{1}$,
  $\log_{10}F_{4}$ for MBC.} 
\label{tab:var:select}
\centering
\begin{tabular}{rlrcr}
\hline\hline
Step & Variable  & Step Type & BIC Difference & Decision\\
\hline
1 & $\log_{10}T_{90}$ & Add & 452.95 & Accepted \\
2 & $\log_{10}T_{50}$ & Add & 395.74 & Accepted \\
3 & $\log_{10}P_{64}$ & Add & 181.35 & Accepted \\
4 & $\log_{10}P_{64}$ & Remove & 181.73 & Rejected \\
5 & $\log_{10}P_{1024}$ & Add & 1636.98 & Accepted \\
6 & $\log_{10}T_{90}$ & Remove & 391.06 & Rejected \\
7 & $\log_{10}P_{256}$ & Add & 1266.77 & Accepted \\
8 & $\log_{10}T_{90}$ & Remove & 235.84 & Rejected \\
9 & $\log_{10}F_{3}$ & Add & 540.03 & Accepted \\
10 & $\log_{10}T_{90}$ & Remove & 243.47 & Rejected \\
11 & $\log_{10}F_{2}$ & Add & 509.78 & Accepted \\
12 & $\log_{10}T_{90}$ & Remove & 95.59 & Rejected \\
13 & $\log_{10}F_{1}$ & Add & 312.38 & Accepted \\
14 & $\log_{10}T_{90}$ & Remove & 45.00 & Rejected \\
15 & $\log_{10}F_{4}$ & Add & 113.55 & Accepted \\
16 & $\log_{10}F_{4}$ & Remove & 16.09 & Rejected \\
\hline
\end{tabular}
\end{table}
Fig.~\ref{fig:BIC} indicates overwhelming evidence in favor of a
five-component $t$MM, with a difference of greater than 10 than for other
$K$, which as per ~\citet{kassandraftery95}, constitutes very strong
evidence.  The results obtained mirror those of
\citet{chattopadhyayandmaitra17} which also found five groups upon
using GMMBC and six parameters.
\begin{figure}
  \centering 
  \mbox{
        \hspace{-0.045\textwidth}
        \subfloat[]{\label{table:BIC}
        \begin{tabular}{cc}
$K$ & BIC\\
\hline
          1 & -11120.1\\  
          2 & -7625.3\\ 
          3 & -6885.8\\  
          4 & -6705.8\\  
          5 & -6532.6\\
          6 & -6588.5\\  
          7 & -6707.0\\  
          8 & -6693.9\\
          9 & -6856.2\\
          \hline
        \end{tabular}
      }
      \centering
      \subfloat[]{\label{fig:bic}
      \raisebox{-.5\height}
      {\includegraphics[width=0.35\textwidth]{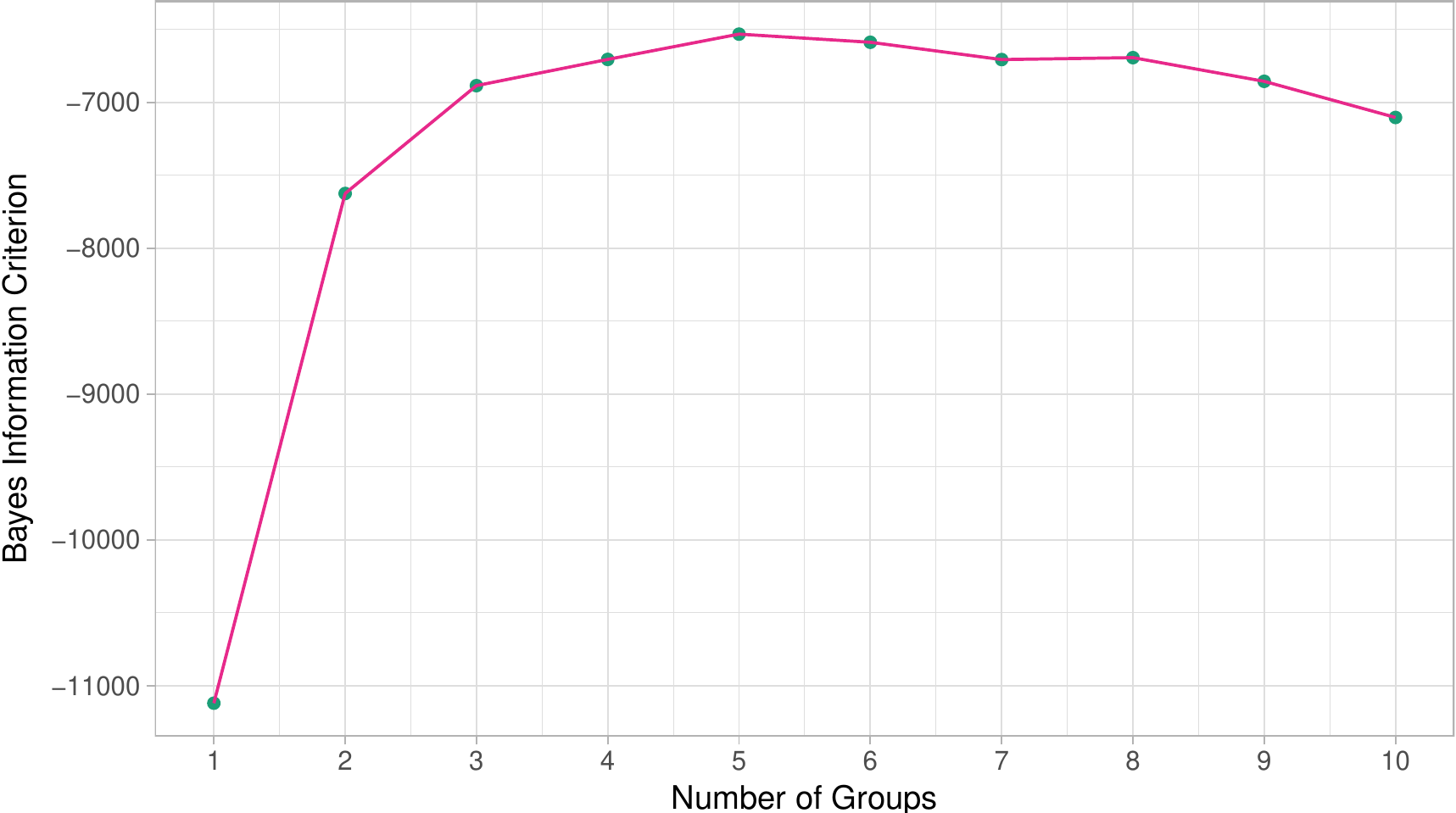}} }
        \centering
      }
      \caption{BIC for each $K$ upon performing $t$MMBC of the 1599
        GRBs in the BATSE 4Br catalogue.}  
      \label{fig:BIC}
    \end{figure}

A reviewer asked why we did not simply eliminate variables that did
not have much apparent additional  information because they were highly
correlated with other variables. We refer back to the example of
Fig.~\ref{sim-norm-var-select} in
Section~\ref{ov:variable selection} and note that there also we had
very high correlations 
(0.98) between two variables, but both were needed to be included with
the third variable for good clustering performance. 
The fact that the correlation between
any pair is high does not necessarily mean that one of the
variables in the pair is redundant for clustering and can simply be
dropped. Indeed, it is possible that there is more redundancy of the
variables with regard to defining some group and not in the case of
others. We return to this point again in Section~\ref{sec:analysis},
but note our preference for using the data to systematically inform us
of relevant and irrelevant variables for clustering. Our formal
variable selection algorithm establishes the relevance of all nine
parameters. Also as mentioned in
Section~\ref{sec:tmmbc}, the {\tt TEIGEN} (and {\tt MCLUST}) family
allows for restricted  dispersion matrices with the
choice governed by BIC that penalizes more complicated (i.e. less
restricted) models. 
\paragraph{Validity of obtained groupings}
 We calculate the empirical pairwise overlap by fitting a GMM as
described in the the {\sc MixSim} package \citep{melnykovetal12}.
The overlap map of Fig. \ref{overlap:teigen} shows the distinctness of the five groups obtained using $t$MMBC.
It is evident that the Group 4 have  very small overlap with both
Groups 2 and 3. On the other hand, Groups 1 and 4
have the highest 
overlap, while the pairwise overlap measures between  Groups  1, 2 and  
3 are moderate. The overlap map indicates that the clusters obtained
are quite well-separated and so our results find five  
GRB sub-populations that are more distinct than \citet{chattopadhyayandmaitra17}. In order to provide a clearer understanding of the
figures in the overlap map we provide a visual representation of the
three pairs of groups having approximate overlaps of 0, 0.1 and 0.05,
using an Andrews plot~\citep{andrews72} in Fig. \ref{andrews-overlap} that
provides an
\begin{figure}
\centering
{\includegraphics[width=0.50\textwidth]{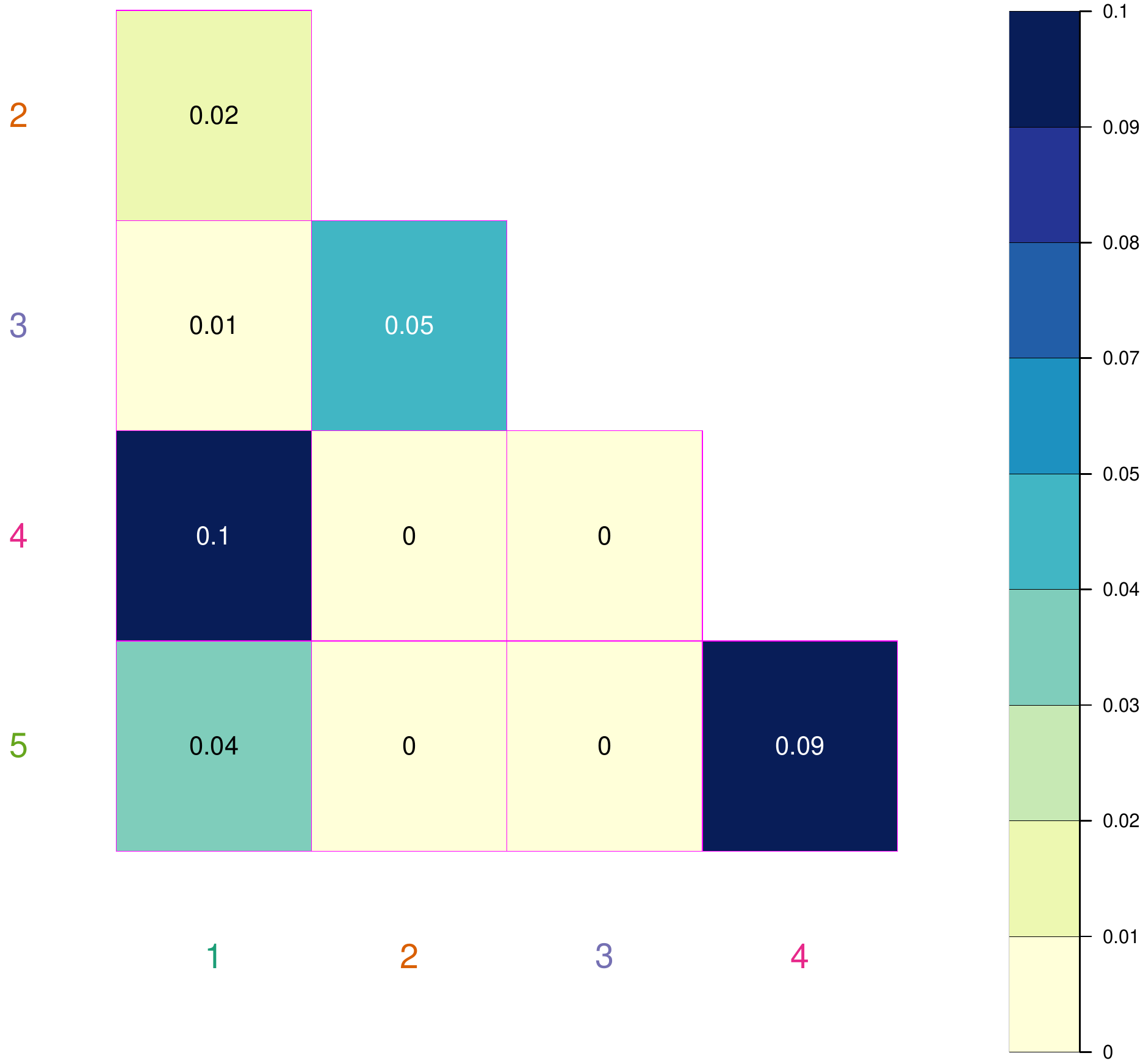}} 
\caption{Pairwise overlap measures between the $k$th and the $l$th  groups obtained by our five-component $t$MMBC solutions.}
\label{overlap:teigen}
\end{figure}
effective way to visualize multivariate data. In Andrews plot, a
realization $x = (x_1, x_2,\ldots, x_p)^T$ 
is represented by a curve $f_x(t)$ in argument $t$ $(-\pi \leq t \leq \pi)$ where $f_x(t)$ defines a finite Fourier series
\begin{equation}
\label{andrews:fourier}
f_x(t) = \frac{x_1}{\sqrt{2}} + x_2 \sin t + x_3 \cos t + x_4\sin 2t
+x_5\cos 2t+ \ldots
\end{equation}
Thus each observation is represented as a curve in $(-\pi, \pi)$. \cite[For a detailed review of Andrews curves,
see][]{khattreeanddayanand02}. In Fig. \ref{andrews:0} the
two groups represented by the curves of different colours are much
more distinct than those in figures \ref{andrews:0.05} and
\ref{andrews:0.1}. Indeed, in none of the cases do the curves of
different colours track together: there is separation at some point or
the other on every curve. Similar comparison curves obtained
from the other pairs of groups show that 
our five groups are well-separated compared to the groups obtained by
\citet{chattopadhyayandmaitra17}. 
\begin{figure*}
\mbox{ 
\subfloat[Groups 2 and 4 GRBs]{\label{andrews:0}\includegraphics[width=0.33\textwidth]{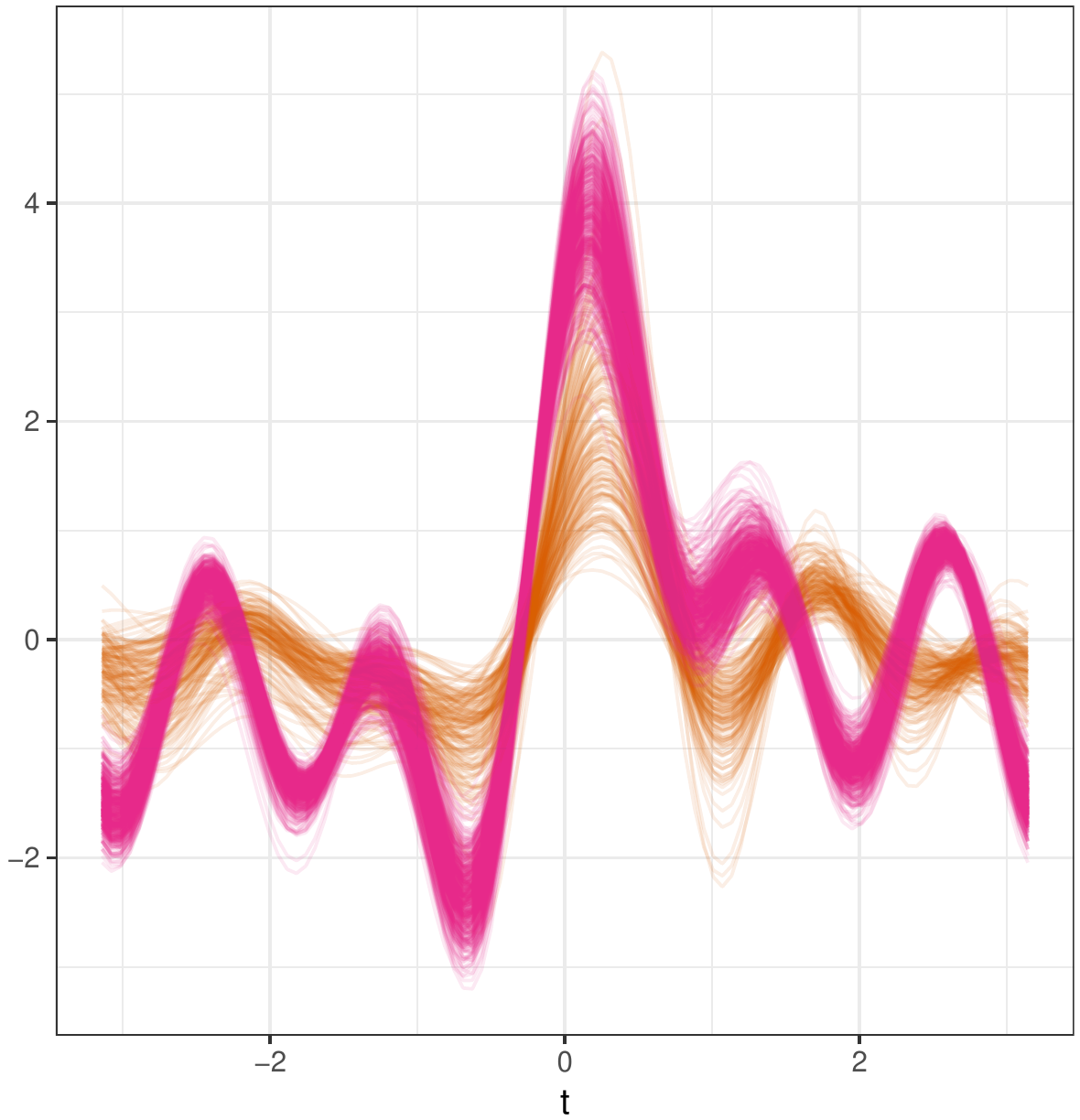}}
\subfloat[Groups 2 and 3 GRBs]{\label{andrews:0.05}\includegraphics[width=0.33\textwidth]{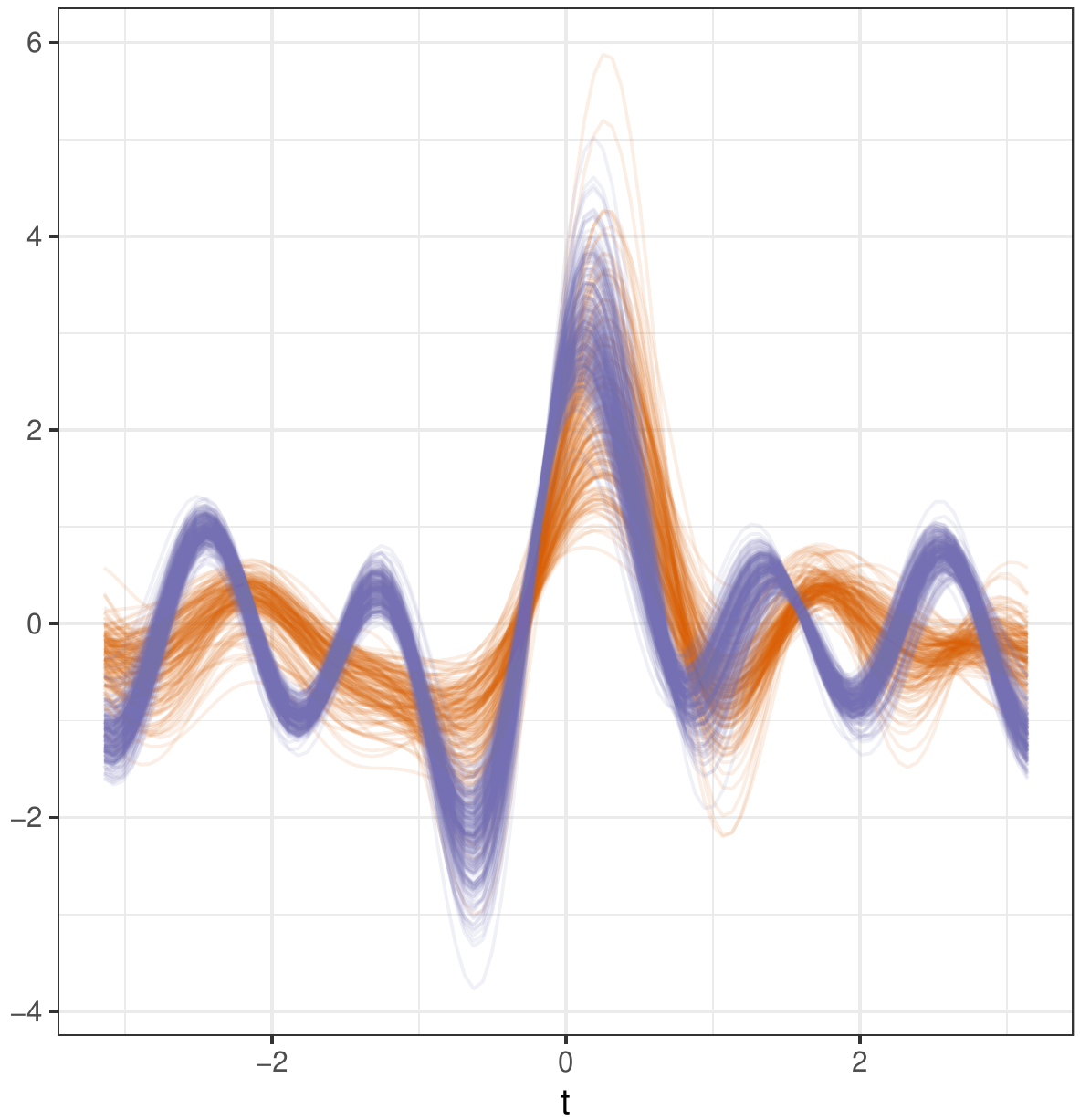}}
\subfloat[Groups 1 and 4 GRBs]{\label{andrews:0.1}\includegraphics[width=0.33\textwidth]{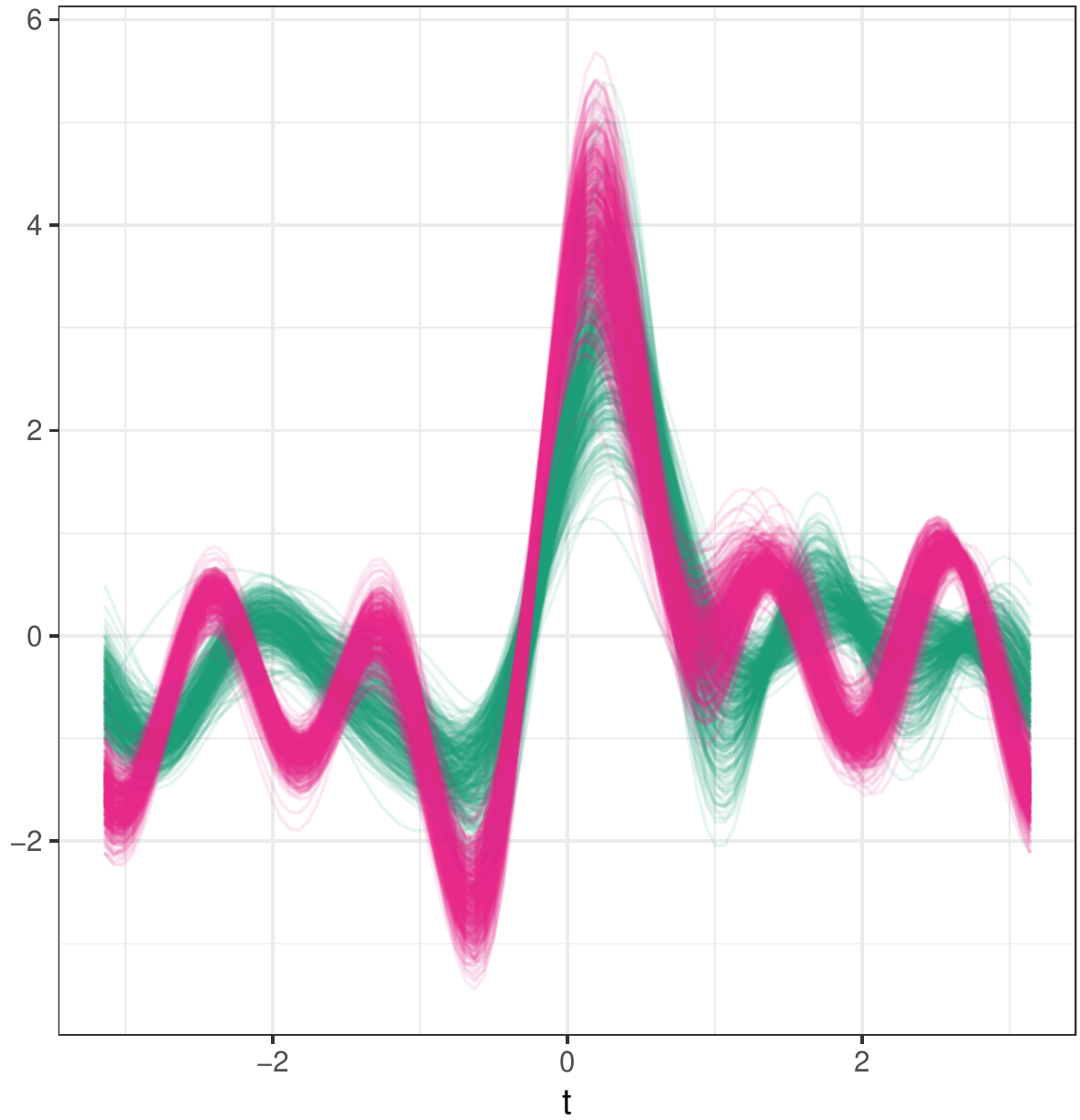}}}
\caption{Andrews plot of observations in (a) Groups 2 and 4 having
  negligible pairwise overlap, (b) Groups 2 and 3 having overlap of
  0.05 and (c) Groups 1 and 4 having overlap of 0.1}
  \label{andrews-overlap}
\end{figure*}
Also, the generalized overlap (see
\citet{chattopadhyayandmaitra17}) for the five component solutions is
0.05 which is much less than the 0.10 obtained for five-component GMMBC solution
in \citet{chattopadhyayandmaitra17} that adds support
to the fact the groups obtained here are more distinct than than
those obtained by \citet{chattopadhyayandmaitra17}. 

A reviewer wondered whether our five-cluster results were by chance
and whether accounting for variability in the estimated groupings and
parameters would yield different, perhaps more conservative,
results. We note that BIC chose five clusters and did so, as per
\citet{kassandraftery95} decisively, with a difference of over 10 than
all other models and components under consideration. To further
investigate the strength of this result, we also used a nonparametric bootstrap
technique to estimate the distribution of the number of kinds of GRBs
in the BATSE catalog. Specifically, we used 1000 bootstrap replicates
of the dataset, with each replicate  obtained by sampling with
replacement 1599 records from the complete dataset. For each
replicate, we fit $t$MMs in the same manner as in Section~\ref{MBC:all} and
used BIC to select the order of the model from among $K=2,3,4,5$. All 1000 bootstrap
replicates chose the five-component cluster model, providing
confidence in our findings that there are five ellipsoidal
sub-populations of GRBs in the BATSE 4Br catalogue.

A second reviewer asked us to clarify that our nonparametric bootstrap
procedure as implemented above only provides an estimate of the number of
clusters that is selected using BIC. It, of course, does not provide a
$p$-value of the hypothesis that $K=5$ is the smallest number of
clusters compatible with the data. To address this latter question,
we could implement a parametric bootstrap procedure whereby
the bootstrap samples are drawn from the $K_*$-component $t$-mixture
distribution fitted to the original data. The $p$-value of the test of
$H_0: K=K_*$ versus the alternative hypothesis $H_a: K=K^*$ could be 
approximated calculated on the basis of the bootstrap replications of
the likelihood ratio test statistic formed for each bootstrap sample
after the fitting of $K_*$ and $K^*$ $t$-component mixture densities
densities to it. (For testing the competing hypothesis of two against
five groups, for instance, we would have $K_*=2$ and $K^*=5$ in the
above specification. For testing that $K=5$ is the smallest number of
groups compatible with the data, we would set $K_*=5$ and
$K^*= K_1>5$.) Indeed, this is a very
computationally intensive procedure requiring multiple
initializations and fittings for each of the pairs of models posited
in the two competing hyptheses.

\paragraph{Analysis of Results}\label{sec:analysis}
 
   \begin{figure}
   \includegraphics[width=0.5\textwidth]{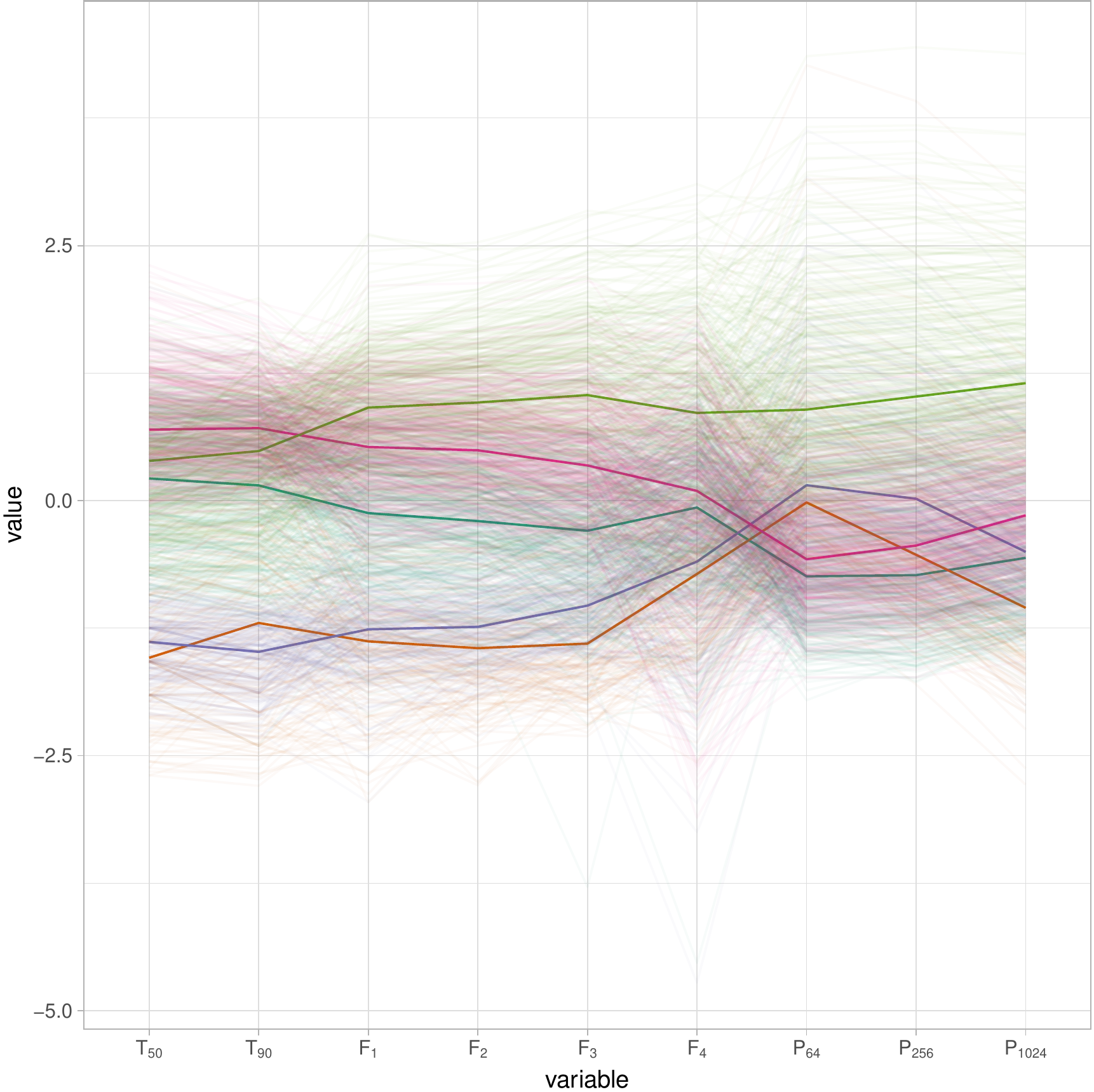}
   \caption{Parallel coordinate plot of the 1599 BATSE 4Br GRBs colored as per
     their group indicators. The  solid lines represent the group medians for each of  the nine variables displayed. Variables are in the logarithmic scale.}
\label{pcp.teigen.plot}
\end{figure}

\begin{table*}
  \caption{Number of GRBs in each of the five groups obtained using $t$MMBC.}     
  \label{tab:nks:teigen}
  \centering
  \begin{tabular}{rrrrrr} \hline\hline
    Group & {\color{Gr1} 1} & {\color{Gr2} 2} &  {\color{Gr3} 3} & {\color{Gr4} 4} &  {\color{Gr5} 5} \\ \hline
    Number of observations & 360 & 160 & 237 & 479 & 363 \\ \hline
  \end{tabular}
\end{table*}
\begin{table*}
  \caption{Mean (top row) and median (bottom row) parameter values for each of
    the five $t$MMBC groups. The figures in parenthesis are (top row)
    the standard error of the mean and (bottom row) the group inter-quartile
    range of the observations in that parameter.}
  \label{tab:means:teigen}
  \begin{tabular}{rrrrrrrrrr}           \hline\hline
        $k$ &$\log T_{50}$&$\log T_{90}$&$\log F_{1}$&$\log F_{2}$&$\log  F_{3}$&$\log F_{4}$&$\log P_{64}$&$\log P_{256}$&$\log P_{1024}$\\ \hline\hline
    \color{Gr1} 1&0.72(0.03)&  1.10(0.03)& -6.87(0.03)& -6.76(0.03)& -6.30(0.03)& -5.97(0.04)& 0.11(0.02)& 0.00(0.02)& -0.16(0.02)\\
   & 0.76(0.92) & 1.15(0.95) & -6.85(0.82) & -6.76(0.63) & -6.29(0.67) & -5.94(0.87) & 0.04(0.44) & -0.07(0.46) & -0.20(0.46)\\        
           \color{Gr2} 2&-0.62(0.07)& -0.07(0.07)& -7.98(0.05)& -7.77(0.04)& -7.05(0.04)& -6.42(0.05)& 0.43(0.03)&  0.13(0.03)& -0.36(0.04)\\
   & -0.86(1.06) & -0.09(1.25) & -8.00(0.75) & -7.82(0.62) & -7.15(0.68) & -6.49(0.93) & 0.36(0.36) & 0.02(0.46) & -0.44(0.60)  \\      
           \color{Gr3} 3&-0.74(0.02)& -0.37(0.02)& -7.90(0.03)& -7.61(0.02)& -6.82(0.02)& -6.49(0.05)& 0.50(0.03)&  0.32(0.02)& -0.13(0.02)\\
   & -0.72(0.41) & -0.35(0.48) & -7.90(0.49) & -7.64(0.46) & -6.86(0.44) & -6.39(0.83) & 0.43(0.51) & 0.26(0.44) & -0.18(0.44)\\        
           \color{Gr4} 4&1.24(0.02)&  1.67(0.02)& -6.27(0.02)& -6.18(0.02)& -5.78(0.02)& -5.86(0.03)& 0.13(0.01)&  0.07(0.01)&  0.01(0.01)\\
    & 1.20(0.62) & 1.66(0.51) & -6.25(0.60) & -6.17(0.58) & -5.80(0.56) & -5.80(0.81) & 0.12(0.31) & 0.06(0.34) & 0(0.35)\\       
           \color{Gr5} 5&0.88(0.03)&  1.43(0.03)& -5.91(0.03)& -5.76(0.03)& -5.27(0.03)& -5.17(0.04)& 0.82(0.02)&  0.78(0.02)&  0.69(0.02)\\
    & 0.92(0.68) & 1.46(0.66) & -5.90(0.71) & -5.77(0.69) & -5.26(0.81) & -5.14(1.14) & 0.75(0.57) & 0.72(0.57) & 0.63(0.56)\\       
           \hline
  \end{tabular}
\end{table*}

Table \ref{tab:nks:teigen} provides the number of observations in each
group, with the color for the group indicators matching the  color of
the groups in all figures to provide for easy cross-referencing.
We see that Groups 1 and 4 contain the highest number of GRBs while
Group 2 contains the lowest 
number of GRBs. Table \ref{tab:means:teigen} also lists the estimated means of the
five groups. The standard errors of the estimated means for the five groups are also provided in the parenthesis corresponding to each estimate. A more detailed visual representation is provided by 
Fig. \ref{pcp.teigen.plot} which displays the five groups via a
parallel coordinate plot\citep{inselberg85,wegman90, chattopadhyayandmaitra17}. 

Most authors have used the duration variable $T_{90}$ to describe the
group properties of GRBs while some have used fluences $F_1 - F_4$
along with duration. \citet{mukherjeeetal98} used total fluence $F_t =
F_1 + F_2 + F_3 + F_4$ and hardness ratio $H_{321} = F_3/F_1 + F_2$
along with  $T_{90}$ to describe the properties of their groups, a scheme that
was also adopted by \citet{chattopadhyayandmaitra17}. We follow this scheme to describe our groups using the
three properties duration-total Fluence-Spectrum. Using this rule the
five groups of \citet{chattopadhyayandmaitra17} were
intermediate-faint-intermediate, long-intermediate-soft,
intermediate-intermediate-intermediate, short-faint-hard and
long-bright-intermediate. This same  paradigm classifies our five
$t$MMBC-obtained groups as long-intermediate-intermediate,
short-faint-intermediate, short-faint-soft, long-bright-hard and long-intermediate-hard.

For a further study of the five groups, we take a closer
look at the duration variable $\log_{10}T_{90}$. This variable also
facilitates comparison of our results to those obtained in
\citet{chattopadhyayandmaitra17} and other authors such as
\citet{mukherjeeetal98} who have used $\log_{10}T_{90}$ along with
other derived variables such as $\log_{10}F_t$. Our fourth and fifth
group contains the bursts of highest duration (around $47s$ and $27s$
respectively). Typically the bursts from these two groups and the
first group are designated as long-duration bursts ($T_{90} > 2s$)
following the popular classification scheme of classifying bursts with
duration less than $2 s$ as short duration bursts and bursts greater
than $2 s$ as long duration bursts~\citep{chattopadhyayetal07}. The
second and third groups  consist of bursts of shortest duration
(around $0.4s$ and $0.7s$ respectively) and will be classified as
short duration bursts ($T_{90} < 2 s$).  
The standard errors of the
estimated means are not large and show that the estimated means of the
groups are distinct. 

We also compared our $t$MMBC grouping with the GMMBC grouping
of \citet{chattopadhyayandmaitra17} by means of a cross classification table (Table \ref{tab:cross}).
\begin{table}
  \caption{Number of 1599 GRBs assigned to each of the groupings by
    $t$MMBC
using the nine original variables (Grouping I) and GMMBC of \citet{chattopadhyayandmaitra17} using three original variables ($\log_{10}T_{50}$, $\log_{10}T_{90}$, $\log_{10}P_{256}$) and three derived variables $\log_{10}F_t$, $\log_{10}H_{32}$, $\log_{10}H_{321}$.}
\label{tab:cross}
  \begin{tabular}{cc|ccccc|c}\hline\hline
    && \multicolumn{6}{c}{Grouping I (New groupings)} \\
    &&\color{Gr1} 1&\color{Gr2} 2&\color{Gr3} 3&\color{Gr4} 4&\color{Gr5} 5&Total\\ \hline
    \multirow{6}{*}{\begin{sideways}{Grouping II}\end{sideways}}
    &\color{Gr1} 1&86&  50&  14&  24&   0&174\\
    &\color{Gr2} 2&5&  57&  25&   2&  60&149\\
    &\color{Gr3} 3&45&  48& 198&   0&   1&292\\
    &\color{Gr4} 4&186&   5&   0& 319&  41&551\\
    &\color{Gr5} 5& 38&   0&   0& 134& 261&433\\ 
    \hline
    &Total&360& 160& 237& 479& 363 \\ 
  \end{tabular}
  \end{table}
The high values in the diagonal (with the exception for Group 2)
indicates that the grouping structure in both the cases agrees well
for both the analyses with the highest agreement being noted in the
Group 4. 

\begin{figure*}
\mbox{ 
\subfloat[Group 1]{\label{fig:cor1.teigen}\includegraphics[height=0.28\textheight]{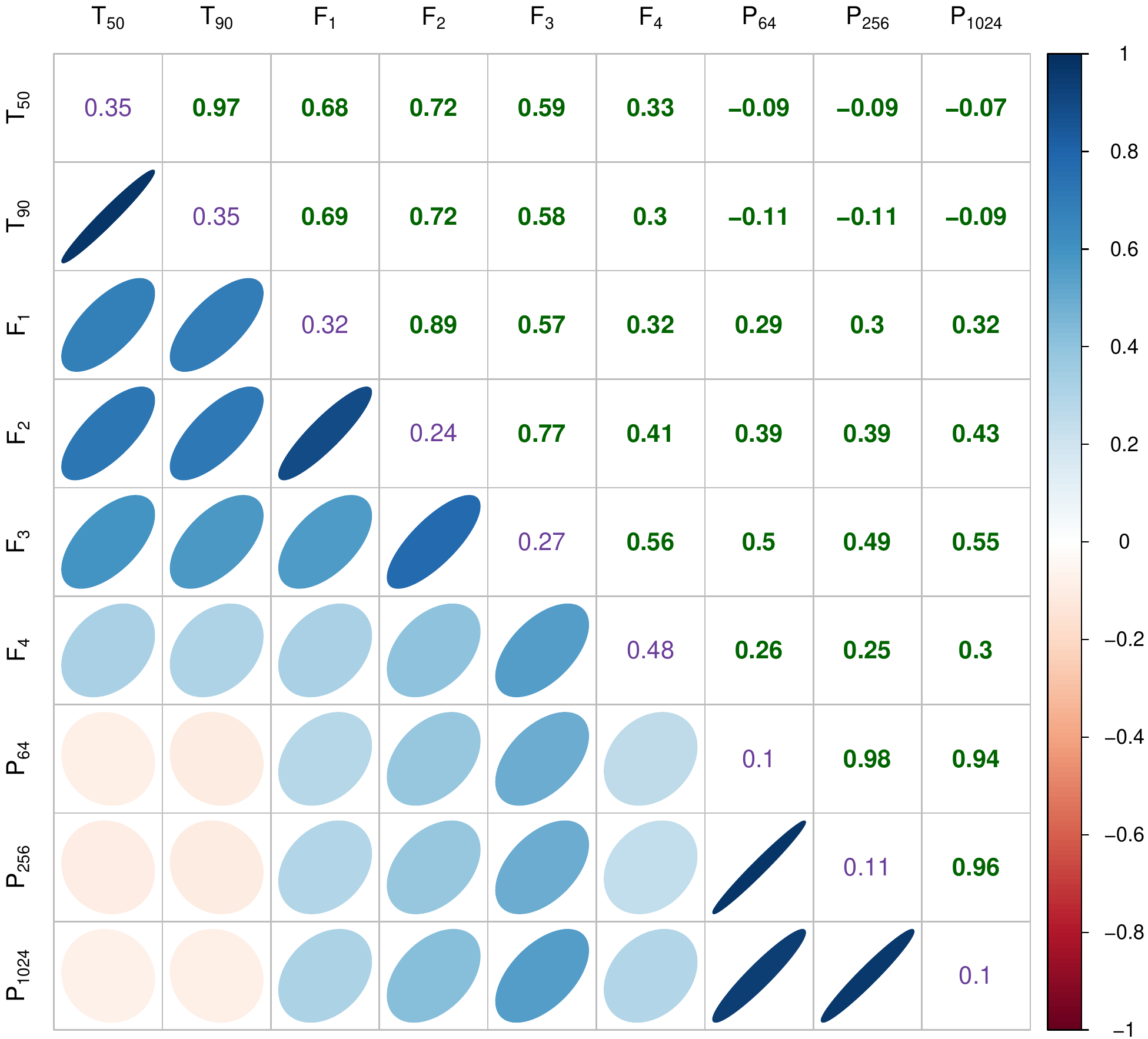}}
\subfloat[Group 2]{\label{fig:cor2.teigen}\includegraphics[height=0.28\textheight]{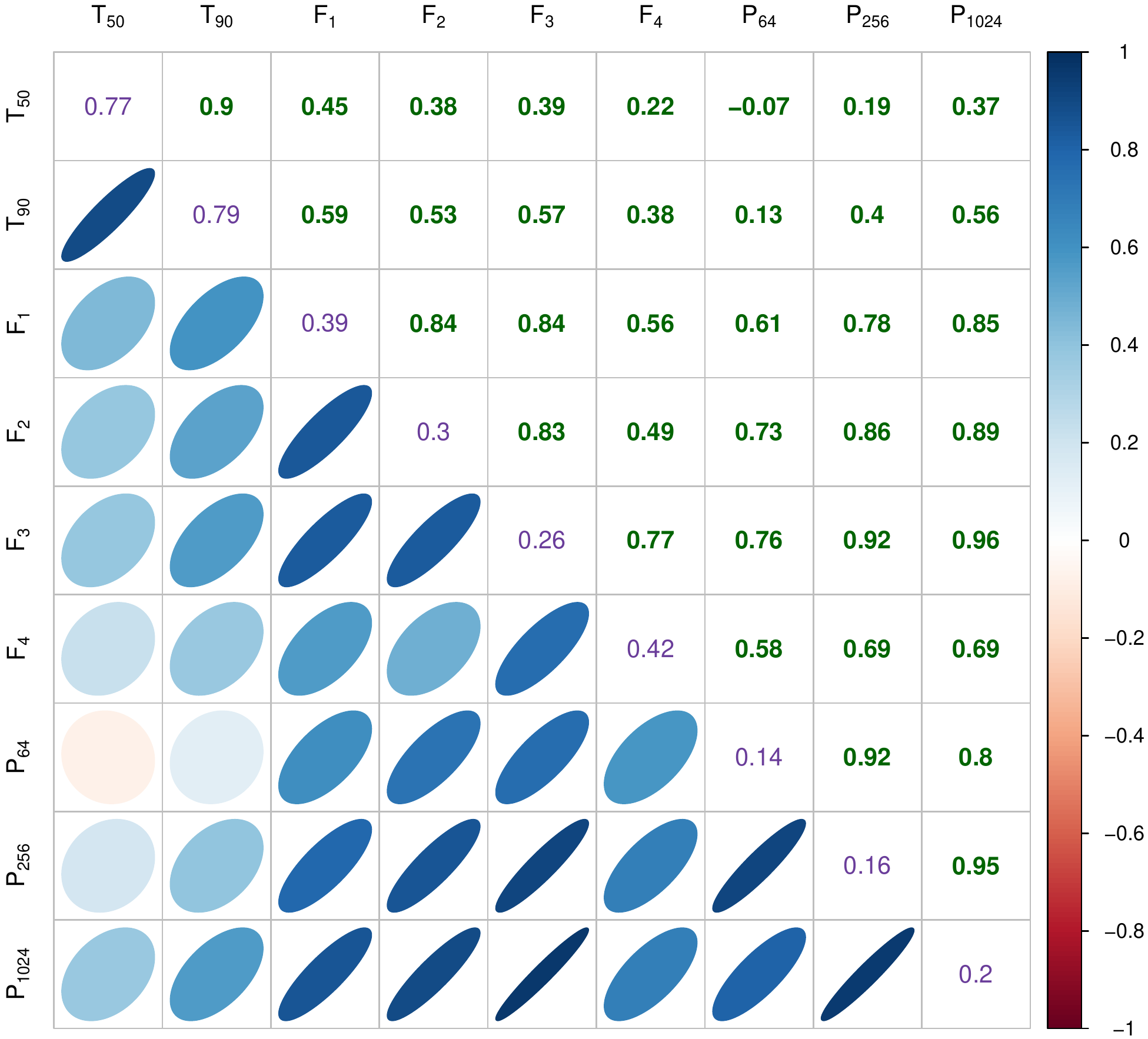}}}
\mbox{
\subfloat[Group 3]{\label{fig:cor3.teigen}\includegraphics[height=0.28\textheight]{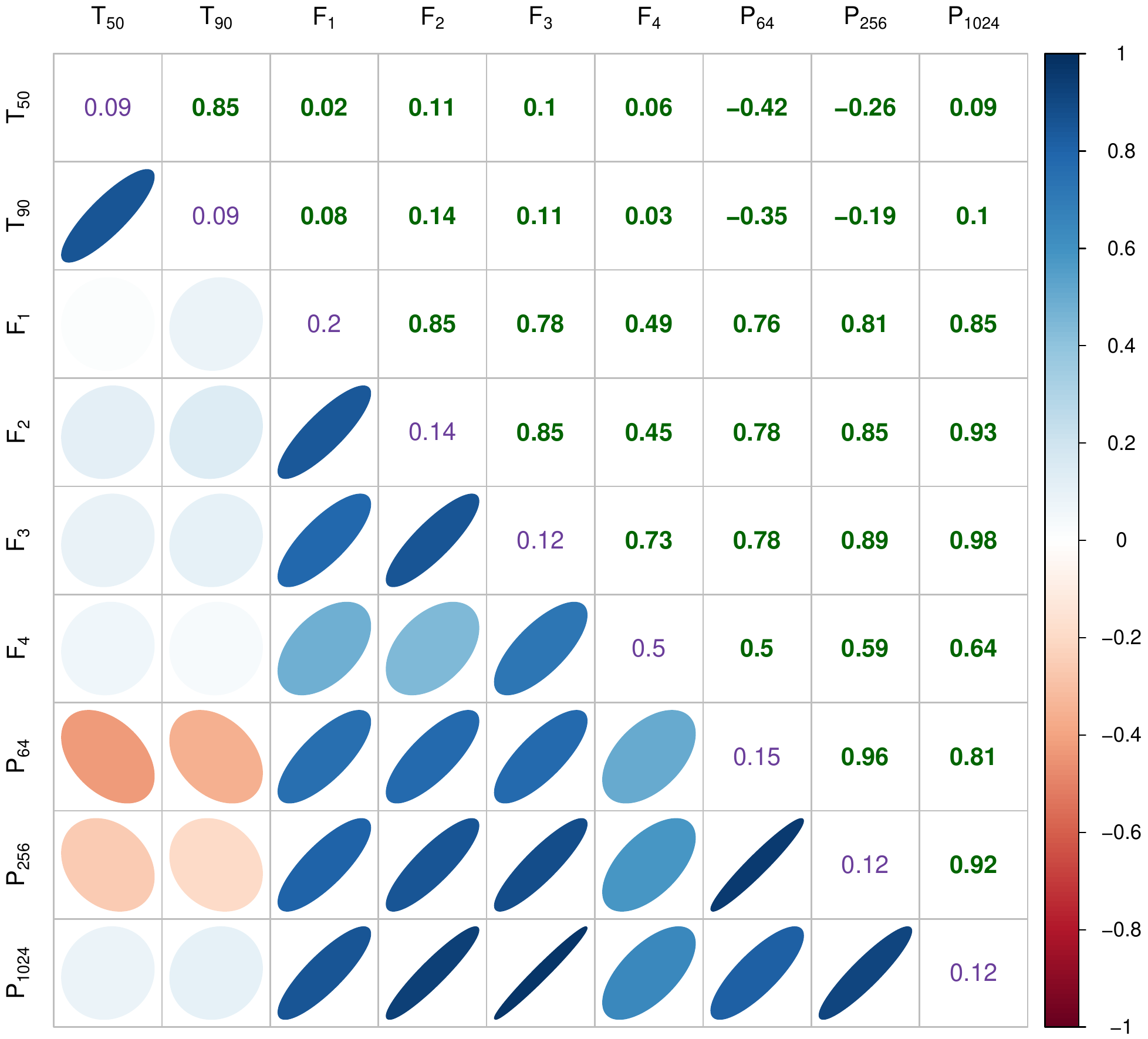}}
\subfloat[Group 4]{\label{fig:cor4.teigen}\includegraphics[height=0.28\textheight]{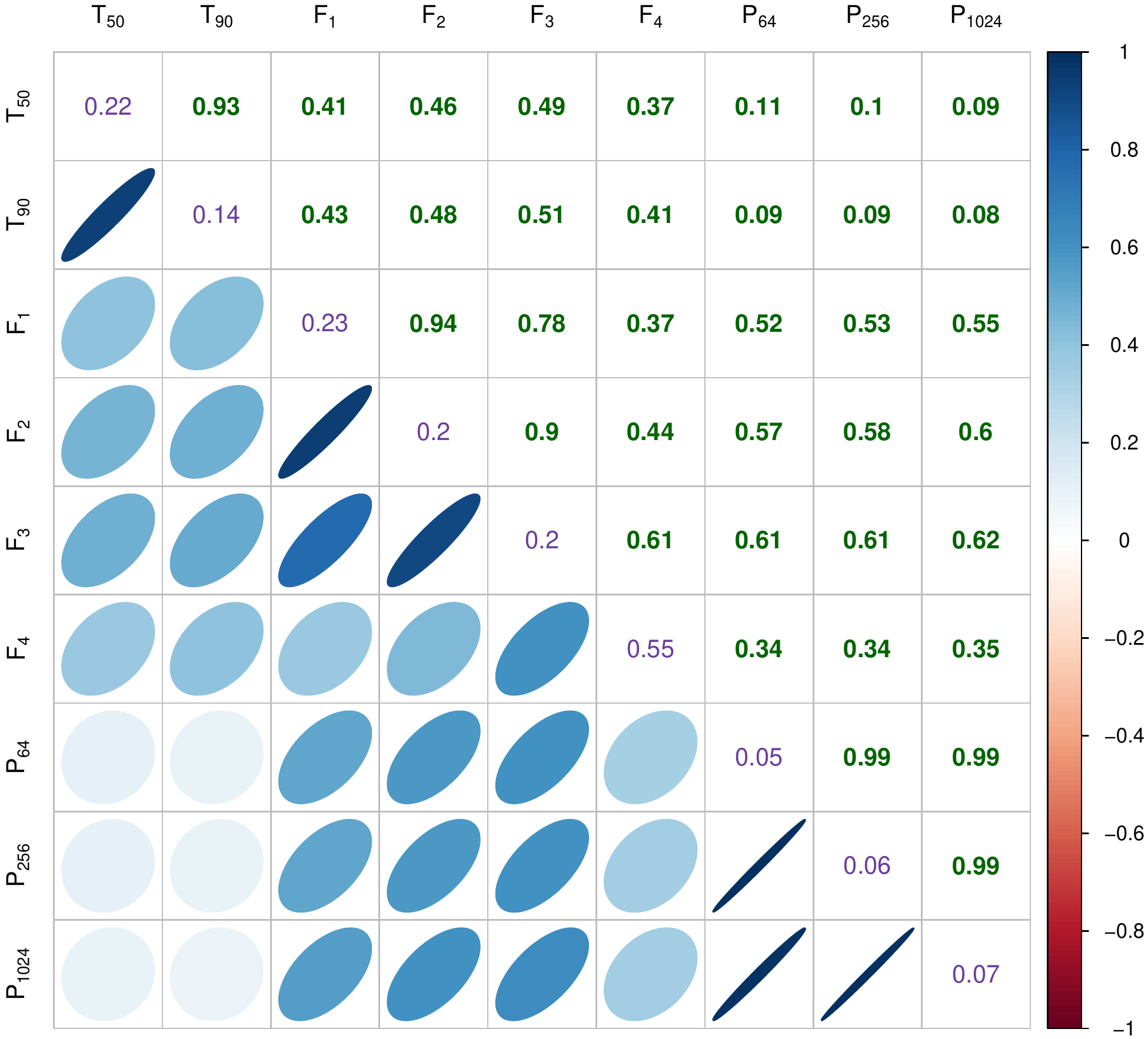}}}
\subfloat
    [Group 5]{\label{fig:cor5.teigen}\includegraphics[height=0.28\textheight]{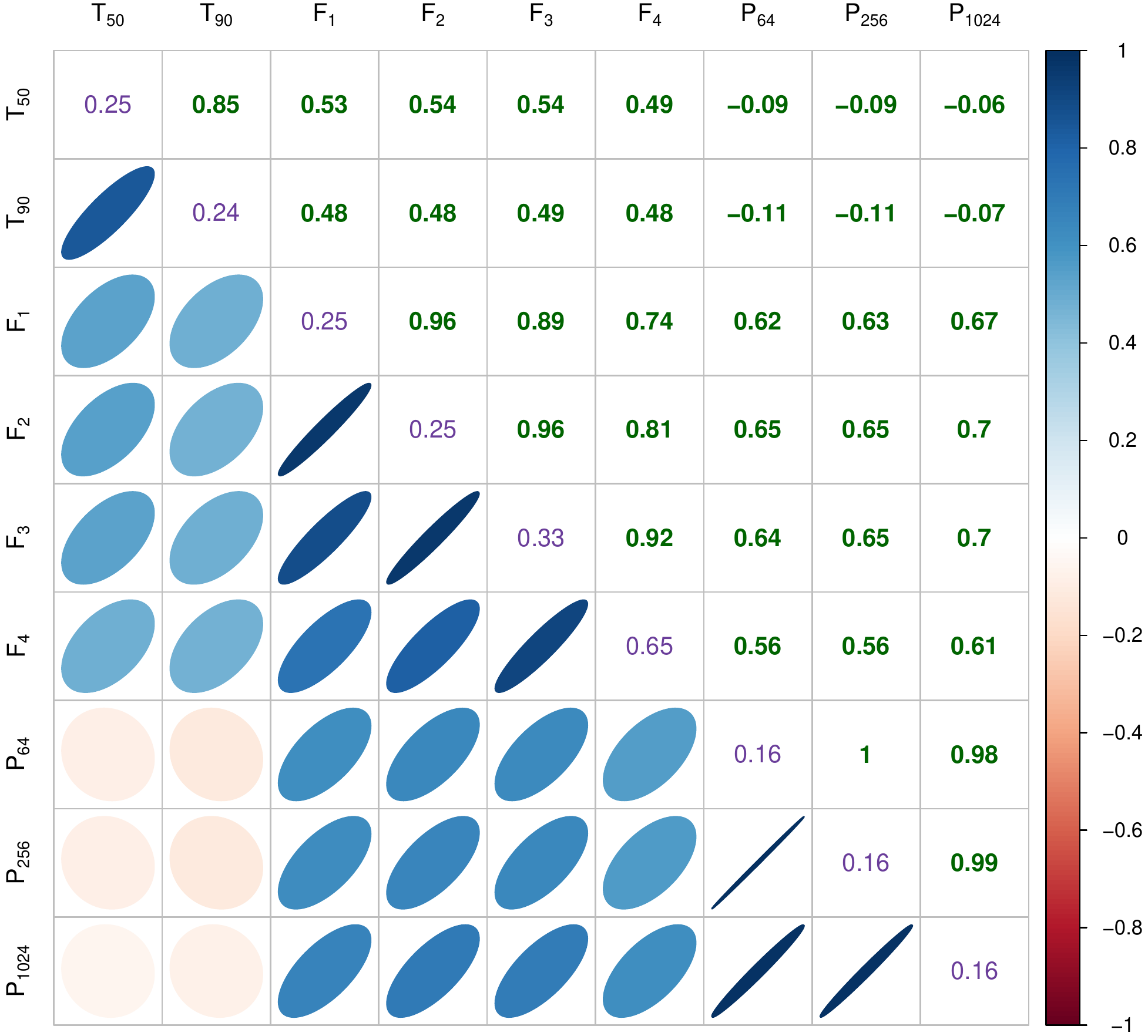}}
\caption{Variances and displays of the estimated correlations for each
  of the five   groups obtained from the five-component MBC solution using {\it t} mixtures for the 1599 GRBs. For each group, the off-diagonal elements
  display correlation between the variables while the diagonals
  display the variances. Both correlations and variances are
  calculated for the variables in the base-10 logarithmic scale.}
  \label{corrplot.6var.teigen}
\end{figure*}

In order to facilitate further study of the group structures we
calculated the correlations between the nine classes for each of the
five groups (Fig. \ref{corrplot.6var.teigen}). The diagonals of each
correlation plot display the estimated variance of the five
groups obtained by $t$MMBC. The upper triangular
portion of each correlation plot displays the correlation
between the variables, while the lower triangular part provides a
diagrammatic representation of the correlations in the upper
triangular part [e.g. the (2,1)th cell displays the correlation
between $\log_{10}T_{50}$ and $\log_{10}T_{90}$ diagrammatically, while
the (1,2)th cell displays the numerical value]. The duration variables
$\log_{10}T_{50}$ and $\log_{10}T_{90}$ display very high positive
association in all five groups. Duration $\log_{10}T_{50}$ and
fluence $\log_{10}F_1$ have a high positive association in Group 1 and moderate positive association in Groups 2, 4 and 5. In Group 3, they
have very low positive association. The other fluences
$\log_{10}F_2$, $\log_{10}F_3$ and $\log_{10}F_4$ also exhibit similar
linear relationships with $\log_{10}T_{50}$ for the five
groups except that $\log_{10}F_4$ shows a moderate positive association in Group 1. Duration $\log_{10}T_{90}$ has a high positive association
with $\log_{10}F_1$ in Group 1. In Groups 2, 4 and 5 they display a moderate positive association and a weak positive association in Group 3 . Also, $\log_{10}T_{90}$ displays a moderately positive
association with fluence $\log_{10}F_2$ in Groups 2, 4 and 5 and high positive
association in Group 1. In Group 3, they 
display a low positive association. Fluence $\log_{10}F_4$ exhibits a
moderate positive association with $\log_{10}T_{90}$ in Groups 4 and 5 and a weak positive association in Groups 1 and 2. They have a very
weak positive association in Group 3. Fluences $\log_{10}F_1$ and
$\log_{10}F_2$ show a very high positive association in all five
groups, $\log_{10}F_2$ and $\log_{10}F_3$ have a strong positive
association in Groups 1 and 2 while in Groups 3-5 they have a very strong
positive association. Fluence $\log_{10}F_4$ has a moderate positive
association with $\log_{10}F_1$ in Groups 2, 3 and 4 and a strong
positive association in Group 5. In Group 1 they exhibit a weak
positive association. Peak flux $\log_{10}P_{64}$ and
$\log_{10}T_{90}$ have a weak positive association in groups 2 and
4. In Groups 1, 3 and 5 they show a weak negative association. These
correlations are in agreement with the hypothesis that the total amount of fluence in higher duration
bursts is likely to more compared to that of shorter duration bursts.  

It is important to note that the variables that had very high
correlations in Fig.~\ref{fig:Bivariate9} have different correlations
in different groups. For instance, $\log T_{50}$ and $\log T_{90}$
have a correlation of 0.967 in the entire dataset, but ranges from
0.85 to 0.97 depending on the group. The correlation structures for
the different groups indicate that it may be possible to have a
lower-dimensional representation for some of them. Indeed, a factor
analysis \citep{johnsonandwichern88} of the 
observations in each of the groups indicated that four factors (but
not parameters) may adequately explain the relationship between the
parameters in Group 4, but not for the other groups. Therefore, it is appropriate to allow for general
dispersion structures for all the parameters in our $t$MMBC. 

Our analysis so far has been on 1599 GRBs for which observations are
available for all nine parameters. We now use the results
of our $t$MMBC on 1599 GRBs to classify the $374$ BATSE GRBs with
incomplete observations. 

\subsection{GRBs with partially observed parameters}
\label{class-374}
\subsubsection{Descriptive Analysis}
There are $374$ GRBs in the BATSE catalogue with incomplete
information in one or more parameters (mainly in the four fluences
$F_1 - F_4$), as seen in Table \ref{tab:nks.miss} that enumerates
the number of missing observations for in each of the nine parameters.
We present in Fig. \ref{fig:violin.plot} a split violinplot \citep{hintzeandnelson98}
with the left side of the violin displaying the distribution of the parameters
from the 1599 GRBs and the right side displaying the distribution
of the 374 GRBs that are missing observations in any other parameter. The violin plots for the two duration variables $\log_{10}T_{50}$ and
$\log_{10}T_{90}$ are very similar but the four peak fluxes $\log_{10}F_1
- \log_{10}F_4$ generally have lower values for the missing
cases. The
densities of the three peak fluxes $\log_{10}P_{64}$,
$\log_{10}P_{256}$ and $\log_{10}P_{1024}$ have heavier right tails for the
1599 GRBs compared to those of the 374 GRBs with incomplete
observation. The 374 GRBs with 
missing parameter  observations show a good degree of symmetry in all the three
peak fluxes. Consequently, and as expected, their values also are
generally lower for the observations with missing parameters than for
the observations with all parameters observed.

\subsubsection{Classification}
In Section~\ref{MBC:all} we 
excluded these GRBs from multivariate analysis since
standard MBC  techniques and available software are not suited to address situations with
missing variables. Here, we illustrate how we can use the clustering
results of` Section~\ref{MBC:all} along with classification
methods of Section~\ref{ov:classification} to group the GRBs with
missing parameters. We first develop some
methodology for this purpose.

\begin{table}
  \caption{Number ($n_j$) of observations with incomplete information in each of
    the BATSE 4Br catalogue parameters (denoted by $X_j$).}  \label{tab:data:missing}
    \begin{tabular}{cccccccccc}
\hline\hline
$X_j$&$T_{50}$&$T_{90}$&$P_{64}$&$P_{256}$&$P_{1024}$&$F_1$&$F_2$&$F_3$&$F_4$\\
\hline
$n_j$&0&0&1&1&1&29&12&6&339\\\hline\hline
\end{tabular}
\label{tab:nks.miss}
\end{table}

Corollary \ref{corr:t} implies that excluding the parameters  that are
missing for a GRB yields an observation of reduced dimensions that
still has a multivariate {\it t}-distribution with parameters 
corresponding to the observed parameters. Therefore the
classification rule of Section~\ref{ov:classification} can still be
used with reduced $t_\nu$ density as per Corollary \ref{corr:t} taking the place of the densities $f_i$'s in that section. Therefore, the
parameter estimates returned by $t$MMBC as per our ECM algorithm of
\ref{ov:teigen} can be used. Specifically, the estimated $\mu_k$s are
used, but only the parameters that are observed for the GRB under
consideration are included in the calculation of the classification
rule. Similarly, only the rows and columns of $\Sigma_k$s that correspond
to the observed parameters are included in the classification rule. The estimated prior
proportions $\pi_k$'s can be used unchanged in the classification rule
calculations.  
\begin{table*}
\caption{(a) Number of missing observations in each of the nine parameters for each of the five groups. (b) group means of the nine parameters for the classified
     GRBs. Note that this mean is computed for a parameter only for
     those GRBs that did not have missing observations in that
     parameter for that group.}
\label{tab:missing}
     \mbox{\subfloat[Number of missing observations in each of the parameters for each group]{\label{tab:nks:class}{
            \centering
            \begin{tabular}{rrrrrrrrrr}          \hline\hline
            	$k$ &$\log T_{50}$&$\log T_{90}$&$\log F_{1}$&$\log F_{2}$&$\log
            	F_{3}$&$\log F_{4}$&$\log P_{64}$&$\log P_{256}$&$\log P_{1024}$\\ \hline
            	\color{Gr1} 1 & 0 & 0 & 6 & 6 & 5 & 131 & 1 & 1 & 1 \\
            	\color{Gr2} 2 & 0 & 0 & 11 & 4 & 1 & 38 & 0 & 0 & 0\\ 
            	\color{Gr3} 3 & 0 & 0 & 10 & 1 & 0 & 22 & 0 & 0 & 0\\
            	\color{Gr4} 4 & 0 & 0 & 1 & 1 & 0 & 125 & 0 & 0 & 0\\
            	\color{Gr5} 5 & 0 & 0 & 1 & 0 & 0 & 23 & 0 & 0 & 0\\
            	\hline
            \end{tabular}
          }
     }
   }
  
\mbox{\subfloat[Mean (top row) and median (bottom row) parameter
  values for each group. The figures in parenthesis are (top row) the standard
  error of the mean and the inter-quartile range (bottom row). For
  each group, calculations are based on the GRBs that are not missing
  the particular parameter for that group. ``-'' indicates that there
  was only one observed field for that parameter in that group.]{\label{tab:means:class}{
      \hspace{-0.05\textwidth}
		\begin{tabular}{rrrrrrrrrr}          \hline\hline
           $k$ &$\log T_{50}$&$\log T_{90}$&$\log F_{1}$&$\log F_{2}$&$\log
           F_{3}$&$\log F_{4}$&$\log P_{64}$&$\log P_{256}$&$\log
                                                             P_{1024}$\\
                  \hline\hline
           \color{Gr1} 1 & 0.68(0.05) & 1.05(0.05) & -6.81(0.04) & -6.76(0.04) & -6.50(0.04) & -5.95(0.26)  & -0.04(0.02) & -0.16(0.02) & -0.30(0.02)\\
       & 0.67(0.87) &  1.05(0.88) & -6.82(0.55) & -6.75(0.56) &
                                                                -6.44(0.58) & -5.96(0.53) & -0.07(0.23) &  -0.21(0.25) & -0.34(0.232)\\ \hline
\color{Gr2} 2 & -0.76(0.11) & -0.37(0.11) & -7.92(0.07) & -7.95(0.05) & -7.34(0.04) & -6.62(0.16) &  0.19(0.03) & -0.11(0.02) & -0.61(0.03)\\
     & -0.89(0.88) & -0.55(1.06) & -8.03(0.59) & -7.91(0.44) &
                                                               -7.32(0.35) & -6.53(0.75) & 0.18(0.30) & -0.104(0.279) & -0.593(0.314)\\ \hline
\color{Gr3} 3& -0.74(0.05) & -0.35(0.05) & -8.10(0.08) & -7.84(0.04) & -7.17(0.04) & -7.09(0.32) & 0.18(0.04) & -0.01(0.04) & -0.44(0.04)\\
    & -0.74(0.30) & -0.35(0.48) & -8.07(0.29) & -7.83(0.28) &
                                                              -7.17(0.21) & -6.81(0.96) &  0.17(0.20) & 0.02(0.23) & -0.43(0.18)\\ \hline
\color{Gr4} 4& 1.23(0.04) &  1.63(0.03) & -6.34(0.04) & -6.29(0.04) & -6.02(0.03) & -8.07(1.49)  &  0.01(0.02) & -0.07(0.02) & -0.13(0.02)\\
   & 1.26(0.55) &  1.66(0.40) & -6.33(0.53) & -6.29(0.45) &
                                                            -6.01(0.51) & -8.07(1.49) &  -0.01(0.27) & -0.10(0.32) & -0.18(0.33)\\ \hline
\color{Gr5} 5& 0.89(0.11) &  1.53(0.11) & -6.32(0.09) & -6.18(0.10) & -5.91(0.12) & -3.95(-) &  0.36(0.07) &  0.31(0.07) & 0.18(0.07)\\
  & 0.85(0.80) &  1.60(0.80) & -6.33(0.51) & -6.18(0.58) & -6.03(0.55)
                  & -3.96(0) & 0.35(0.33) & 0.29(0.33) & 0.15(0.35)\\ 
           \hline\hline
         \end{tabular}
       }
     }
   }
      \end{table*}
\begin{figure}
\includegraphics[width=0.5\textwidth]{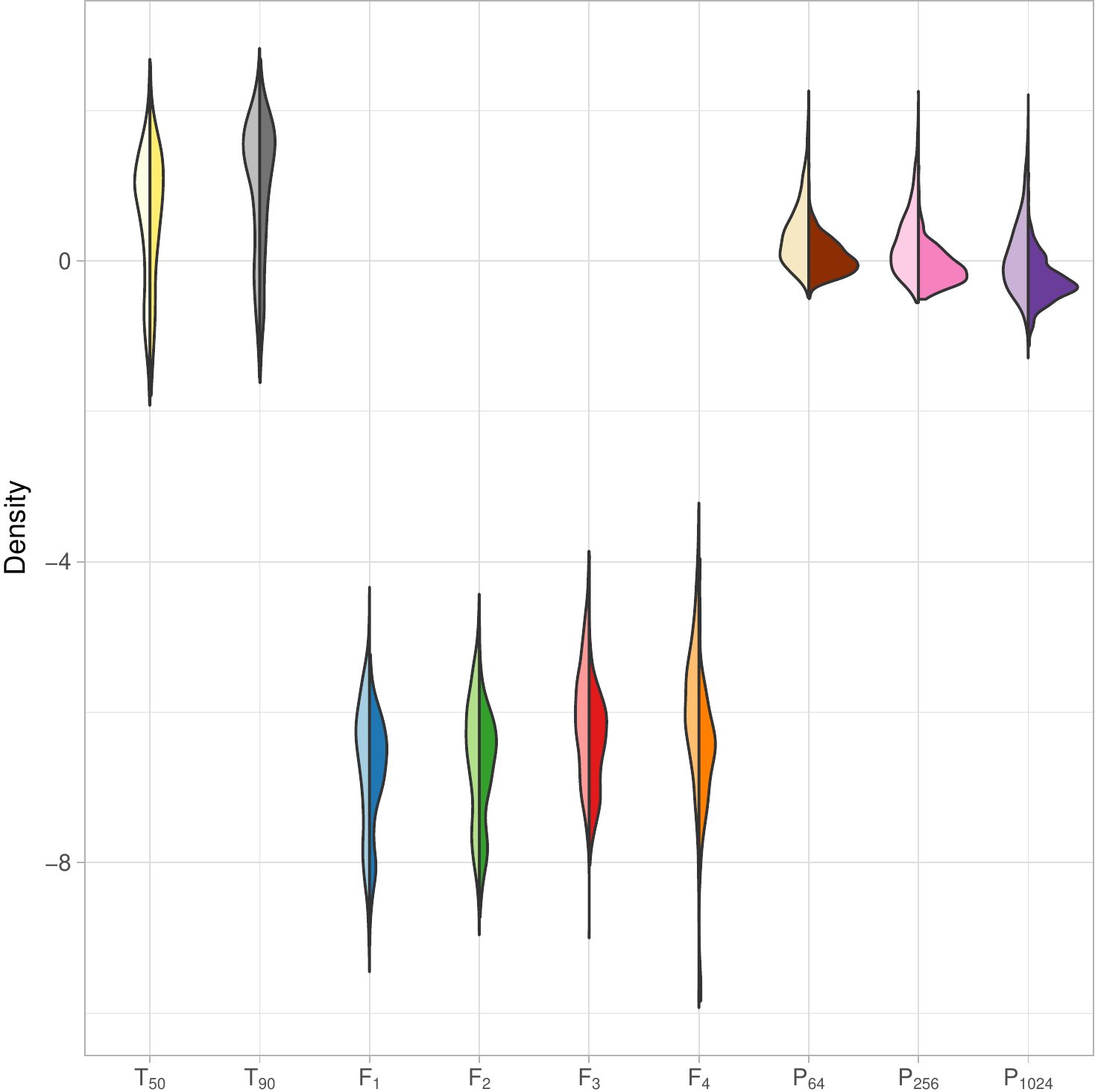}
\caption{Split violin plot of the nine observed variables where the
  left side of each violin is the kernel density estimate of the 1599 GRBs
  with complete information and the right side is the kernel density estimate
  of the 374 GRBs with incomplete information. } 
\label{fig:violin.plot}
\end{figure}

Using the above rule, we classified the 374
GRBs. Table \ref{tab:missing} displays the number of missing
parameters in each group and the descriptive statistics for each
group. 
There were 138 GRBs classified to the long-intermediate-intermediate group, 52 GRBs to short-faint-intermediate group, 33 GRBs to short-faint-soft group, 127 GRBs to long-bright-hard group and 24 GRBs to long-intermediate-hard group.
 Finally, we note that our 
classification strategy does not have the ability to find
(potentially) additional  classes
in the GRBs with missing parameters because we are using the classes
found from clustering the 1599 GRBs in the assignment. Therefore, it
would be desirable to have methodology that groups GRBs with missing
and complete observations in a holistic approach. Development of such
a statistical approach for {\it t}MMBC, while necessary, is however beyond the scope
of this paper.  

\section{Conclusions}
\label{conclusion}
Many authors have attempted to determine the kinds 
of GRBs in the BATSE catalogue using various statistical
techniques. While most authors have suggested that there are two kinds of GRBs, a few others have claimed this number to 
to be three and not two. Recently
\citet{chattopadhyayandmaitra17} classified 1599 GRBs using GMMBC
using three original and three derived variables and found the
optimal number of groups to be five. They presented carefully analyzed
evidences in support of their findings. Motivated by the fact that the nine
original variables might contain useful clustering information we
carried out MBC using the nine original variables after checking for
redundancy among them. Clustering the 1599 GRBs using $t$MMBC
showed that the optimal number of homogeneous  groups is
five and further supports the results obtained by
\citet{chattopadhyayandmaitra17} providing evidence that the
additional ellipsoidal groups found by them can not be subsumed inside
groups with heavier tails.  These groups are also more distinct than
the the groups obtained in \citet{chattopadhyayandmaitra17} as per
the overlap measure of \citet{maitraandmelnykov10}. Using the classification scheme of \citet{mukherjeeetal98} and \citet{chattopadhyayandmaitra17} our five groups have were classified as long-intermediate-intermediate, short-faint-intermediate, short-faint-soft, long-bright-hard, and long-intermediate-hard. Further a Bayes Classifier categorized 374 GRBs having missing information in one or more of the parameters to the five groups obtained from the 1599 GRBs having complete information using $t$MMBC. 138 GRBs were classified to the long-intermediate-intermediate group, 52 GRBs to the short-faint-intermediate group, 33 GRBs to the short-faint-soft group, 127 GRBs to the long-bright-hard group and 24 GRBs to the long-intermediate-hard group.

Our article has found five ellipoidally-dispersed groups. Recent 
work~\citep{almodovarandmaitra18} on syncytial clustering when applied
to the results of the analysis indicated that
the number of general-shaped groups in the BATSE 4Br GRB catalog is
indeed five and these five groups happen to be
ellipsoidally-dispersed. Therefore, we have great confidence in our
finding that there are five kinds of GRBs in the BATSE 4Br catalog.

There are a number of issues which can be looked upon as  potential
research problems. For one, it would be useful to incorporate and
further develop clustering methods that have the ability to group
observations that are complete and missing information in a holistic
manner. \citet{lithioandmaitra18} have, among others, redesigned the
$k$-means algorithm for such  scenarios but efficient methodology and software 
to fit $t$-mixture models with incomplete records would also be
helpful. Further, use of the logarithmic transformation, while  
standard in GRB analysis, may obfuscate further group
structure so an approach which incorporates finding the
transformation within the context of clustering would be worthwhile to
explore. Finally, the analysis in this paper can be
extended to GRBs catalogued  from sources such as the datasets from
the \textit{Swift} and \textit{Fermi} satellites to analyse whether
similar  results hold for GRBs observed from these other satellites.   

\section*{Acknowledgments}
We sincerely thank the Editor-in-Chief and two reviewers whose
insightful comments greatly improved this manuscript. Thanks also to
Asis Kumar Chattopadhyay and Tanuka Chattopadhyay for originally
introducing us to this research problem.




\bibliographystyle{mnras}
\newpage
\bibliography{references} 

\bsp	
\label{lastpage}
\end{document}